\documentclass{article}%[12pt]
\usepackage{amssymb,latexsym,amsfonts,amsmath,mathrsfs,amsthm,graphicx}
\usepackage{subfigure}
\usepackage{geometry}
\usepackage{setspace}
\usepackage{empheq}
\usepackage{placeins}
\usepackage{accents}
\usepackage{euscript}
\usepackage{hyperref}
\usepackage{tabularx}
\usepackage{ctable}
\usepackage[usenames,dvipsnames]{xcolor}

\usepackage[pagewise]{lineno}
\usepackage{epsf}
\usepackage{epsfig}

\usepackage{bm}
\hypersetup{colorlinks=true, citecolor=blue}
\newtheorem{theorem}{Theorem}
\newtheorem{proposition}[theorem]{Proposition}

\newcommand{\ba}{\begin{array}}
\newcommand{\ea}{\end{array}}

\newcommand{\RR}{{ {\rm I} \! \kern -.05em {\rm R} }}

\begin{document}
\bibliographystyle{abbrv}
\title{From regional pulse vaccination to global disease eradication: insights from a mathematical model of Poliomyelitis
}
%\title{Pulse Vaccination in a Disease Metapopulation Model with Application to Polio Eradication}
%\title{Pulse Vaccination in a Disease Metapopulation Model:  Application to Polio Eradication}
\date{}
\author{Cameron Browne$^1$, Robert Smith?$^2$, Lydia Bourouiba$^3$}
\maketitle

\noindent
1. Department of Mathematics, Vanderbilt University, Nashville TN, cameron.j.browne@vanderbilt.edu

\noindent
2. Department of Mathematics and Faculty of Medicine, The University of Ottawa, Ottawa ON, \\
rsmith43@uottawa.ca 
(to whom correspondence should be addressed)

\noindent
3. Department of Mathematics, Massachusetts Institute of Technology, Boston MA, lbouro@mit.edu

%\linenumbers

\begin{abstract}

\noindent
Mass-vaccination campaigns are an important strategy in the global fight 
against poliomyelitis and  measles.  The large-scale logistics required for these mass immunisation campaigns magnifies the need for research into the effectiveness and optimal deployment of pulse vaccination.  In order to better understand this control strategy, we propose a mathematical model accounting for the disease dynamics in connected 
regions, incorporating seasonality, environmental reservoirs 
 and independent periodic pulse vaccination schedules in each region.  The effective reproduction number, $R_e$, is defined and proved to be a global threshold for persistence of the disease.  Analytical and numerical calculations show the importance of synchronising the pulse vaccinations in connected regions and the timing of the pulses with respect to 
the pathogen circulation seasonality. Our results indicate that it may be crucial for mass-vaccination programs, such as national immunisation days, to be synchronised across different regions.  In addition, simulations show that a migration imbalance can increase $R_e$ and alter how pulse vaccination should be optimally distributed among the patches, similar to results found with constant-rate vaccination.  Furthermore, contrary to the case of constant-rate vaccination, the fraction of environmental transmission affects the value of $R_e$ when pulse vaccination is present.

 \end{abstract}

\section{Introduction}

There are three criteria for the eradication of an infectious disease: 1.~biological and technical feasibility; 2.~costs and benefits; and 3.~societal and political considerations \cite{AlywardHennessey}. 
Despite eradication hopes for malaria, yaws and yellow fever in the twentieth century, smallpox remains the only human disease eradicated \cite{AlywardHennessey}. Current eradication programs include poliomyelitis (polio) \cite{WHOpolio}, leprosy \cite{KealeySmith?} and guinea worm disease \cite{GWD}. Measles, rubella, and hepatitis A and B are also biologically and technically feasible candidates for eradication \cite{Losos}. 
Despite strong biological, technical and cost-benefit arguments for infectious-disease eradication, securing societal and political commitments is a substantial challenge \cite{AlywardHennessey}.

With communities more connected than ever, the control or eradication of an infectious disease requires coordinated efforts at many levels, from cities to nations.  For vaccine-preventable diseases, public health authorities plan immunisation strategies across varying regions with limited resources.  The World Health Organization (WHO) has helped to organise global immunisation efforts, leading to significant global reduction in polio and measles cases \cite{WHOpolio}.  
One vaccination strategy that has been utilised in the global fight against polio and measles is mass immunisation, which may be regarded as a pulse vaccination \cite{NID}.  The complex logistics required for these mass-immunisation campaigns magnifies the need for research into the effectiveness and optimal deployment of pulse vaccination \cite{ZipurskyBoualam}.

Pulse vaccination has been investigated in several mathematical models, often in disease models with seasonal transmission.  Many diseases show seasonal patterns in circulation; thus inclusion of seasonality may be crucial.  Agur et al. (1993) argued for pulse vaccination using a model of seasonal measles transmission, conjecturing that the pulses may antagonise the periodic disease dynamics and achieve control 
at a reduced cost of vaccination \cite{agur}.  Shulgin et al. (1998) investigated the local stability 
of the disease-free periodic solution in a seasonally forced population model with three groups: susceptible (S), infected (I) and recovered (R). They considered 
 pulse vaccination and explicitly found the threshold pulsing period \cite{shulgin}.   Recently, Onyango and M\"{u}ller considered optimal periodic vaccination strategies in the seasonally forced SIR model and found that a well-timed pulse is optimal, but its effectiveness is often close to that of constant-rate vaccination \cite{optimalvac}.  

In addition to seasonality, spatial structure has been recognised as an important factor for disease dynamics and control \cite{ChinaHIV}.  
Heterogeneity in the population movement, along with the patchy distribution of populations, suggests the use of
metapopulation models describing disease transmission in patches or 
spatially structured populations or regions. Mobility can be incorporated and 
tracked in these models in various forms. Common models include 
 linear constant fluxes representing long-term population motion (e.g. migration 
\cite{mig}) and  nonlinear mass-action representing  
short-term mobility \cite{lloyd}.  Liu et al. (2009) and Burton et al. (2012) 
considered epidemic models with both types of movement \cite{burton,liu}.   A possible inherent advantage of pulse vaccination in a spatially structured setting, discussed by Earn et al. (1998), is that the disease dynamics in coupled regions can become synchronised by pulse vaccination, thereby increasing the probability of global disease eradication \cite{Earn}.  Earn et al. presented simulations of patch synchronisation after simultaneous pulse vaccinations in a seasonal SEIR metapopulation model in which the patch population dynamics were initially out of phase. Here an additional population 
class of Exposed (E) was considered.    

Coordinating simultaneous pulse vaccination campaigns in connected regions may be vital for successful employment of pulse-vaccination strategies.  Indeed, employing synchronised pulse vaccinations across large areas in the form of National Immunisation Days (NIDs) and, on an international scale, with simultaneous NIDs, has been successful in fighting polio \cite{NID}.  An example of large-scale coordination among nations is Operation MECACAR (the coordinated poliomyelitis eradication efforts in Mediterranean, Caucasus and central Asian republics), which were initiated in 1995.  The project was viewed as a  success and an illustration of international coordination in disease control \cite{cdc}.  However, public health, including the control of infectious diseases and epidemics, has usually been managed on a national or regional scale, despite the potential impact of population movement \cite{cdc}.    

 Recently, pulse vaccination has been analysed in epidemic metapopulation models \cite{terry, yang}.  Terry (2010) \cite{terry} presented a sufficient condition for eradication in an SIR patch model with periodic pulse vaccinations independently administered in each patch with linear migration rates, but left open the problem of finding a threshold quantity for eradication and evaluating the effect of pulse synchronisation and seasonality.  Yang and Xiao (2010) \cite{yang} conducted a global analysis of an SIR patch model with synchronous periodic pulse vaccinations and linear migration rates; however, they did not allow for the different patches to administer the pulses at distinct times and seasons.  

A major poliovirus transmission route in Africa and the
Middle East is fecal-to-oral transmission. 
In this indirect route, often facilitated by inadequate water management, 
water  plays an analogous role to that of a reservoir, although such an environmental 
reservoir does not allow the pathogen to reproduce. Nevertheless, its effect on the pathogen dispersal   can dramatically 
modify epidemic  patterns 
\cite{BrebanDrake, BTW11}.  The competition 
between the direct and indirect transmission  routes was examined 
in the case of 
highly pathogenic avian influenza H5N1 \cite{BTW11}, showing that indirect
fecal-to-oral transmission could lead to a higher death toll than that
associated with direct contact transmission.

In this article, we consider an SIR metapopulation model with both short- and long-term mobility, direct and indirect (environmental) transmission, seasonality and independent periodic pulse vaccination in each patch.  The primary objectives are to find the effective reproduction number, $R_e$, prove that it provides a sharp eradication threshold and assess the optimal timing of pulse vaccinations in the sense of minimising  $R_e$.   
Our mathematical model and analysis allow us to evaluate how pulse synchronisation across connected patches affects the efficacy of the overall pulse-vaccination strategy.  We also determine how different movement scenarios affect the optimal deployment of vaccinations across the patches, along with considering how environmental transmission affects results. Finally, we discuss how 
pulse vaccination and constant-rate vaccination strategies compare when considering the goal of poliomyelitis global eradication.

This paper is organised as follows.  In Section \ref{sec2}, we describe and give motivation for the mathematical model.  In Section \ref{sec3}, we analyse the disease-free system, which is necessary to characterise the dynamics of the model.  In Section \ref{sec4}, $R_e$ is defined.  In Section \ref{sec5}, we prove that if $R_e<1$, the disease dies out, and if $R_e>1$, then it is uniformly persistent.  In Section \ref{sec6}, we consider a two-patch example, which provides insight into the optimal timing of pulse vaccinations, the effect of mobility and environmental transmission parameters on $R_e$, and a comparison of pulse vaccination to constant-rate vaccination in this setting.  In this section, we also prove that pulse synchronization is optimal for a special case of the model and provide numerical simulations to illustrate this result in more general settings.  Finally, in Section \ref{sec7}, we provide a discussion of the implications of our results and future work to consider.

\section{The mathematical model} \label{sec2}

We consider a variant of an SIR metapopulation model with $N$ patches, each with populations of susceptible, infected and recovered denoted by $S_j$, $I_j$ and $R_j$ for each patch $1\leq j\leq N$.   All three groups migrate from patch $j$ to patch $i$, $i\neq j$, at the rates $m_{ij}S_j$, $k_{ij}I_j$ and $l_{ij}R_j$.  The per capita rates at which susceptible, infected and recovered leave patch $i$ are $m_{ii}=-\sum_{j\neq i} m_{ji}$, $k_{ii}=-\sum_{j\neq i} k_{ji}$ and $l_{ii}=-\sum_{j\neq i} l_{ji}$, respectively.  The effect of short-term mobility on infection dynamics is modelled by mass-action coupling terms; for example, $\beta_{ij}(t)I_jS_i$.  For this infection rate, infected individuals from patch $j$ are assumed to travel to patch $i$, infect some susceptibles in patch $i$ and then return to patch $j$ on a shorter timescale than that of the disease dynamics. Conversely, susceptibles from patch $i$ can travel to patch $j$, become infected and return to patch $i$ on the shorter timescale.

 In the model for poliomyelitis presented herein, both direct contact
 and indirect environmental routes are considered. The environmental
 contamination of the virus in each patch $j$ is described by a
 state variable, denoted by $G_j$.  Infected individuals in patch $j$
 ($I_j$) shed the virus into the environmental reservoir $G_j$ at the
 rate $\xi_j(t)I_j$.  The virus in the environmental reservoir cannot
 reproduce outside of the host and decays at the rate $\nu_j(t)G_j$.
 The virus in the environmental reservoir $j$, $G_j$, contributes to
 the infected population in patch $I_i$ through the mass-action term
 $\epsilon_{ij}(t)S_iG_j$.  Direct transmission contributing to $I_i$
 is represented by the mass-action term $S_i\sum_j \beta_{ij}(t)I_j$.
 Due to possible seasonality of poliovirus circulation
 \cite{WHOpolio}, both direct and environmental transmission
 parameters $\beta_{ij}(t)$, $\epsilon_{ij}(t)$, $\xi_{i}(t)$ and
 $\nu_i(t)$ are assumed to be periodic with a period of one year.

Pulse vaccination is modelled through impulses on the system occurring at fixed times.  First, we consider a general pulse vaccination scheme with no periodicity.  For each patch $i$,  pulse vaccinations occur at times $t_{i,n}$ where $n=1,2,\dots$  At time $t_{i,n}$, a fraction $\psi_{i,n}$ of the susceptible population $S_i$ is instantly immunised and transferred to the recovered class $R_i$.  Therefore $S_i\left((t_{i,n})^+\right)=\left(1-\psi_{i,n}\right)S_i\left((t_{i,n})^-\right)$ and $R_i\left((t_{i,n})^+\right)=\psi_{i,n}S_i\left((t_{i,n})^-\right)$, where $S_i\left((t_{i,n})^+\right)$ and $S_i\left((t_{i,n})^-\right)$ denote limits from the right-hand side and left-hand side, respectively.  

For each patch $i$, demography is modeled with constant birth rate, $b_i$, into the susceptible class and a per capita death rate, $\mu_i$.  The parameter $p_i$ represents the fraction of newborns who are successfully vaccinated.  The parameter $\gamma_i$ is the recovery rate.  Note that both recovery from infection and successful vaccination induce perfect life-long immunity.  All parameters are assumed to be non-negative, and the  parameters $\nu_i(t)$, $\mu_i$ and $b_i$ are additionally assumed to be positive.  We thus arrive at the following mathematical model:
\begin{equation}\label{fmod}
\begin{aligned}
\displaystyle \frac{dS_i}{dt}&=\displaystyle (1-p_i)b_i-\mu_iS_i-S_i\sum_j \beta_{ij}(t)I_j-S_i\sum_j \epsilon_{ij}(t)G_j+\sum_{j}m_{ij}S_j   & t&\neq t_{i,n}\\ 
\displaystyle\frac{dI_i}{dt}&=\displaystyle S_i\sum_j \beta_{ij}(t)I_j+S_i\sum_j \epsilon_{ij}(t)G_j-(\mu_i+\gamma_i)I_i +\sum_{j}k_{ij}I_j   & t&\neq t_{i,n} \\ 
\displaystyle \frac{dG_i}{dt}&=\displaystyle \xi_i(t)I_i-\nu_i(t)G_i  \phantom{+S_i\sum_j \epsilon_{ij}(t)G_j}  & t&\neq t_{i,n}\\ 
\displaystyle \frac{dR_i}{dt}&=\displaystyle p_ib_i+\gamma_iI_i-\mu_iR_i+\sum_{j}l_{ij}R_j    & t&\neq t_{i,n} \\
S_i\left(t_{i,n}^+\right)&=\left(1-\psi_{i,n}\right)S_i\left(t_{i,n}^-\right)  & t&=t_{i,n} \\
R_i\left(t_{i,n}^+\right)&=\psi_{i,n}S_i\left(t_{i,n}^-\right) & t&=t_{i,n}.
\end{aligned}
\end{equation}

Consider the non-negative cone of $\mathbb{R}^{4N}$, denoted by $X=\mathbb{R}^{4N}_+$.  The following theorem shows existence and uniqueness of solutions to (\ref{fmod}), the invariance of $\mathbb{R}^{4N}_+$ and ultimate uniform boundedness of solutions.

\begin{theorem}\label{P1}
 For any initial condition $x_0\in \mathbb{R}^{4N}_+$, there exists a unique solution to system (\ref{fmod}), $\varphi(t,x_0)$, which is smooth for all $t\neq t_{i,n}$ and the flow $\varphi(t,x)$ is continuous with respect to initial condition $x$.  Moreover, the non-negative quadrant $\mathbb{R}^{4N}_+$ is invariant and there exists $M>0$ such that $\limsup_{t\rightarrow\infty}\left\|\varphi(t,x)\right\|\leq M$ for all $x\in \mathbb{R}^{4N}_+$.
\end{theorem}
\begin{proof}
The existence, uniqueness, and regularity for non-impulse times come from results that can be found in \cite{impulsetheory}.  In order to show the invariance of $\mathbb{R}^{4N}_+$, consider the set $\partial Y^0\equiv  \left\{x\in \mathbb{R}^{4N}_+: x_i=0, \  N+1\leq i\leq 3N \right\}$.  On this set, $\frac{dI_i}{dt}=\frac{dG_i}{dt}=0$, so $\partial Y^0$ is invariant.  Also, notice that $\frac{dS_i}{dt}\geq 0$ if $S_i=0$.  Then, by uniqueness of solutions, we find that $\mathbb{R}^{4N}_+$ is invariant.

To show ultimate boundedness, consider the total population of individuals, $N\equiv \sum_i \left(S_i+I_i+R_i\right)$.   Then, adding all the appropriate equations of (\ref{fmod}), we obtain that 
$$\frac{dN(t)}{dt}\leq b-\mu N (t)\qquad \forall t\geq 0, $$
where $b=\max(b_1,\dots,b_N)$ and $\mu=\min(\mu_1,\dots,\mu_N)$.  A simple comparison principle yields $\limsup_{t\rightarrow\infty}N(t)\leq \frac{b}{\mu}$.  This implies that $\limsup_{t\rightarrow\infty}G(t)\leq \frac{\xi b}{\nu \mu}$ where $\xi=\max(\xi_1,\dots,\xi_n)$ and $\nu=\min(\nu_1,\dots,\nu_N)$.  Therefore, if $\varphi(t,x)$ denotes the family of solutions, then there exists $M>0$ such that $\limsup_{t\rightarrow\infty}\left\|\varphi(t,x)\right\|\leq M$ for all $x\in \mathbb{R}^{4N}_+$.
\end{proof}

In order to analyse the asymptotic dynamics of the system, we assume some periodicity in the impulses.  According to WHO guidelines, countries threatened by wild poliovirus should hold NIDs twice a year with 4--6 weeks separating the immunisation campaigns within a year \cite{NID}.  Hence we consider a sufficiently flexible schedule in order to cover this guideline.  Suppose that, for patch $i\in\left\{1,2,\dots,N\right\}$, the pulse vaccinations occur on a periodic schedule of period $\tau_i$. Assume that there exists $\tau\in \mathbb{N}$ such that $\tau =n_1\tau_1=\cdots=n_N\tau_N$, where $n_1,\dots,n_N\in\mathbb{N}$; i.e., there exists a common period $\tau$ for pulse vaccinations among the patches.  For each patch $i\in\left\{1,2,\dots,N\right\}$, we assume that there are $L_i$ pulse vaccinations that occur within the period $\tau_i$.  More precisely, the pulse vaccinations for patch $i$ occur at times $t=n\tau_i+\phi_i^{k}$, where $0 \leq \phi_i^k <\tau_i$, $n\in\mathbb{N}$ and $k\in\left\{1,..,L_i\right\}$.  Note that the recovered (or removed) classes are decoupled from the remaining system and can thus be neglected.  We obtain the following model:
\begin{equation}\label{mod} 
\begin{aligned}
\displaystyle\frac{dS_i}{dt}&=(1-p_i)b_i-\mu_iS_i-S_i\sum_j \beta_{ij}(t)I_j-S_i\sum_j \epsilon_{ij}(t)G_j+\sum_{j}m_{ij}S_j & t&\neq n\tau_i+\phi_i^{k} \\
\displaystyle\frac{dI_i}{dt}&=S_i\sum_j \beta_{ij}(t)I_j+S_i\sum_j \epsilon_{ij}(t)G_j-(\mu_i+\gamma_i)I_i  +\sum_{j}k_{ij}I_j & t&\neq n\tau_i+\phi_i^{k} \\
\displaystyle\frac{dG_i}{dt}&=\xi_i(t)I_i-\nu_i(t)G_i 
& t&\neq n\tau_i+\phi_i^{k} \\
S_i\left((n\tau_i+\phi_i^{k})^+\right)&=\left(1-\psi_i^{k}\right)S_i\left((n\tau_i+\phi_i^{k})^-\right)  & t&=n\tau_i+\phi_i^{k}. 
\end{aligned}
\end{equation}

Model (\ref{mod}) will be analysed in the ensuing sections.

\section{Disease-free system} \label{sec3}
In order to obtain a reproduction number, we need to determine the dynamics of the susceptible population in the absence of infection.  With this in mind, consider the following characterisation of the vaccinations.
Within the time interval $(0,\tau]$, there are $L_1n_1\cdot L_2n_2 \cdots L_Nn_N$ impulses.   Some of these impulses may occur at the same time.  Let $p\leq L_1n_1\cdot L_2n_2\cdots L_Nn_N$ be the number of distinct impulse times.  We label these impulse times in increasing order as follows: $0=t_0\leq t_1< t_2 <\cdots< t_p< t_{p+1}=\tau$.   For $1\leq i \leq N$ and $1\leq \ell \leq p$, let
\begin{align*}
\alpha_i^{\ell}&=
\begin{cases}
   1-\psi_i^k        & \text{if } t_{\ell}=n\tau_i+\phi_i^k  \\
   1        & \text{otherwise.}
  \end{cases}
 \end{align*}
In the absence of infection, we obtain the linear impulsive system:
\begin{align}
\frac{dx(t)}{dt}&=Ax(t)+b & t&\neq n\tau+t_{\ell} & \text{for} \ \ell \in\left\{1,\dots,p\right\}, n\in\mathbb{N}  \label{linear}\\
x\left((n\tau+t_{\ell})^+\right)&=D_{\ell} x\left((n\tau+t_{\ell})^-\right),\nonumber
\end{align}
where
\begin{align*}
x(t)&=\begin{pmatrix} x_1(t) \\ x_2(t) \\ \vdots \\ x_N(t) \end{pmatrix}, \qquad
A = \begin{pmatrix}
  -\mu_1+m_{11}& m_{12} & \cdots & m_{1n} \\
  m_{21} & -\mu_2+m_{22} & \cdots & m_{2n} \\
  \vdots  & \vdots  & \ddots & \vdots  \\
  m_{N1} & m_{N2} & \cdots & -\mu_N+m_{NN}
 \end{pmatrix}, \qquad  b =\begin{pmatrix}(1-p_1)b_1 \\ (1-p_2)b_2 \\ \vdots \\ (1-p_N)b_N \end{pmatrix} \\ \\
D_{\ell}&= {\rm diag}\left(\alpha_1^{\ell},\alpha_2^{\ell},\dots,\alpha_N^{\ell} \right), \qquad \text{and} \qquad m_{ii}=-\sum_{j\neq i} m_{ji}.
\end{align*}

The solution to the linear differential equation $\dot{x}=Ax+b$ (the no-impulse version of (\ref{linear})) is 
\begin{align}
\theta(t,x)=e^{tA}x+(e^{tA}-I)A^{-1}b, \quad \text{where} \ \ \theta(0,x)=x. \label{LinSys}
\end{align}
 We define the period map $F:\mathbb{R}^N\rightarrow\mathbb{R}^N$ for the impulsive system (\ref{linear}).  Here $F(x)=\zeta(\tau,x)$ where $\zeta(t,x)$ is the unique solution of the impulsive system (\ref{linear}) with $\zeta(0,x)=x$.  The interested reader is referred to the work of Bainov and Simeonov \cite{BS1,BS2,BS3} for more information on the theory of impulsive differential equations.
 
 The period map, $F(x)$, can be calculated through recursive relations.  Let $x_0=x$ and iteratively calculate $F(x)$ as follows:
\begin{align}
x_{j}&=D_j\theta(t_j-t_{j-1},x_{j-1}) \qquad \text{for } \ j=1,\dots,p  \label{recursion} \\
F(x)&=\theta(\tau-t_p,x_p).  \notag
\end{align}
Define the following matrices $C_j, \ 1\leq j \leq p+1$ and $C$:
\begin{align*}
C_{p+1}&=I \\
C_p&=e^{(\tau-t_p)A}\cdot D_{p} \\
C_j&=C_{j+1}\cdot e^{(t_{j+1}-t_j)A}\cdot D_j   \quad \text{for }1\leq j\leq p-1 \\
C&=C_1\cdot e^{t_1A}
\end{align*}
Utilising (\ref{LinSys}) to explicitly express the recursive relations in (\ref{recursion}), we obtain the following formula for $F(x)$:
\begin{align}
F(x)&= Cx +\sum_{j=1}^{p} C_{j+1} \left(e^{(t_{j+1}-t_{j})A}-I\right)A^{-1}b. \label{F1} 
\end{align}
This formula will be used to explicitly calculate the periodic solution obtained in the following theorem.  An example of the periodic solution is displayed in Figure \ref{fig:dpu}.

\begin{proposition}\label{lineargas}
The disease-free system (\ref{linear}) has a unique globally asymptotically stable $\tau$-periodic solution $\overline{S}(t)$.
\end{proposition}
\begin{proof}
A $\tau$-periodic solution of (\ref{linear}) corresponds to a fixed point of the period map.  Hence we consider the equation $\overline{x}=F\left(\overline{x}\right)$.
\begin{align*}
\overline{x}&=C\overline{x}+\sum_{j=1}^{p} C_{j+1} \left(e^{(t_{j+1}-t_{j})A}-I\right)A^{-1}b. 
\end{align*}
We claim that $\rho(C)<1$.  To prove this claim, consider the matrix $A$ and define the stability modulus of $A$, $s(A)$, as $s(A)\equiv \max\left\{{\rm Re}(\lambda): \lambda \ \text{is an eigenvalue of } A\right\}$.
Since $A$ is quasi-positive, $A$ has an eigenvalue $\lambda$ such that $\lambda=s(A)$ with an associated non-negative eigenvector $v$ \cite{Thieme}.  Then $w(t)=e^{\lambda t}v$ is a solution to $\dot{x}=Ax$.  Let $\vec{1}=\left(1,1,\dots,1\right)$ be the row vector of ones for $\mathbb{R}^N$.  Define $u(t)=\vec{1}ve^{\lambda t}$.  Then:
\begin{align*}
\frac{du}{dt}&=\vec{1}(\lambda v) e^{\lambda t} =\vec{1}Av e^{\lambda t} =(-\mu_1,-\mu_2,\dots,-\mu_N)ve^{\lambda t}\leq -\mu \vec{1}ve^{\lambda t},
\end{align*}
where $\mu=\min \left\{ \mu_i:  1\leq i \leq N \right\}$. Thus 
$\frac{du}{dt}= -\mu u(t)$, which implies $u(t)\leq \vec{1}ve^{-\mu t}$.  Thus  $\lambda \leq -\mu<0$,  so $s(A)\leq -\mu$.  Therefore, by the standard theory of linear differential equations, there exists a constant $K>0$ such that $\left\|e^{tA}\right\|<Ke^{-\mu t}$.  Then
\begin{align*}
\left\|C^n\right\|&\leq \left\| C\right\|^n  \\
                           & \leq \left(\left\|D_p\right\| || e^{(\tau-t_p)A}|| \left\| D_{p-1} \right\| \cdots \cdots \left\|e^{t_1A}\right\| \right)^n \\
                           &=\left(  || e^{(\tau-t_p)A}|| \prod_{i=1}^p ||e^{(t_i-t_{i-1})A}||\left\|D_i\right\|  \right)^n \\
                           & \leq (K^p e^{-\mu\tau})^n, \quad \text{since $\left\|D_i\right\|\leq 1$ and $||e^{(t_i-t_{i-1})A}||<Ke^{-\mu (t_i-t_{i-1})}$ for each $i$.}
\end{align*}
Thus  $\left\| C^n \right\| \rightarrow 0$ as $n\rightarrow\infty$.  By a standard equivalence result for linear discrete systems, we conclude that $\rho(C)<1$. It follows that the matrix $(I-C)$ is invertible. 
 Hence there is a unique fixed point $\overline{x}$ of the function $F(x)$, given by
\begin{align*}
\overline{x}&=(I-C)^{-1}\sum_{j=1}^{p} C_{j+1} \left(e^{(t_{j+1}-t_{j})A}-I\right)A^{-1}b.
\end{align*}
It follows that $\overline{S}(t)=\zeta(t,\overline{x})$ is a $\tau$-periodic solution of the impulsive system (\ref{linear}).

To show that $\overline{S}(t)$ is globally asymptotically stable, consider the solution $\zeta(t,x)$ to (\ref{linear}) with initial condition $x\in\mathbb{R}^N$.  Suppose $t\in[n\tau,(n+1)\tau)$; i.e., $t=n\tau+s$ for $s\in[0,\tau)$.  Then
\begin{align*}
\left|\zeta(t,x)-\overline{S}(t)\right|&=\left| \zeta(s,F^n(x))-\zeta(s,F^n(\overline{x})) \right| \\
&= \left| e^{(s-t_k)A}D_ke^{(t_k-t_{k-1})A}D_{k-1}\cdots \cdots e^{t_1A}\cdot (C)^n \left(x-\overline{x}\right) \right| \qquad \text{where} \ k\in\left\{0,1,2,\dots,p\right\}  \\
&\leq K \left|(C)^n \left(x-\overline{x}\right) \right|  \quad \text{for some constant } K>0.
\end{align*}
Since $(C)^n\rightarrow 0$ as $n\rightarrow\infty$, the above inequality implies that $\zeta(t,x)\rightarrow \overline{S}(t)$ as $t\rightarrow\infty$.  The above inequality also shows local stability of $\overline{S}(t)$.  Therefore $\overline{S}(t)$ is globally asymptotically stable for the linear impulsive system (\ref{linear}).
\end{proof}

\begin{figure}[t]
\centering
\includegraphics[width=9.5cm]{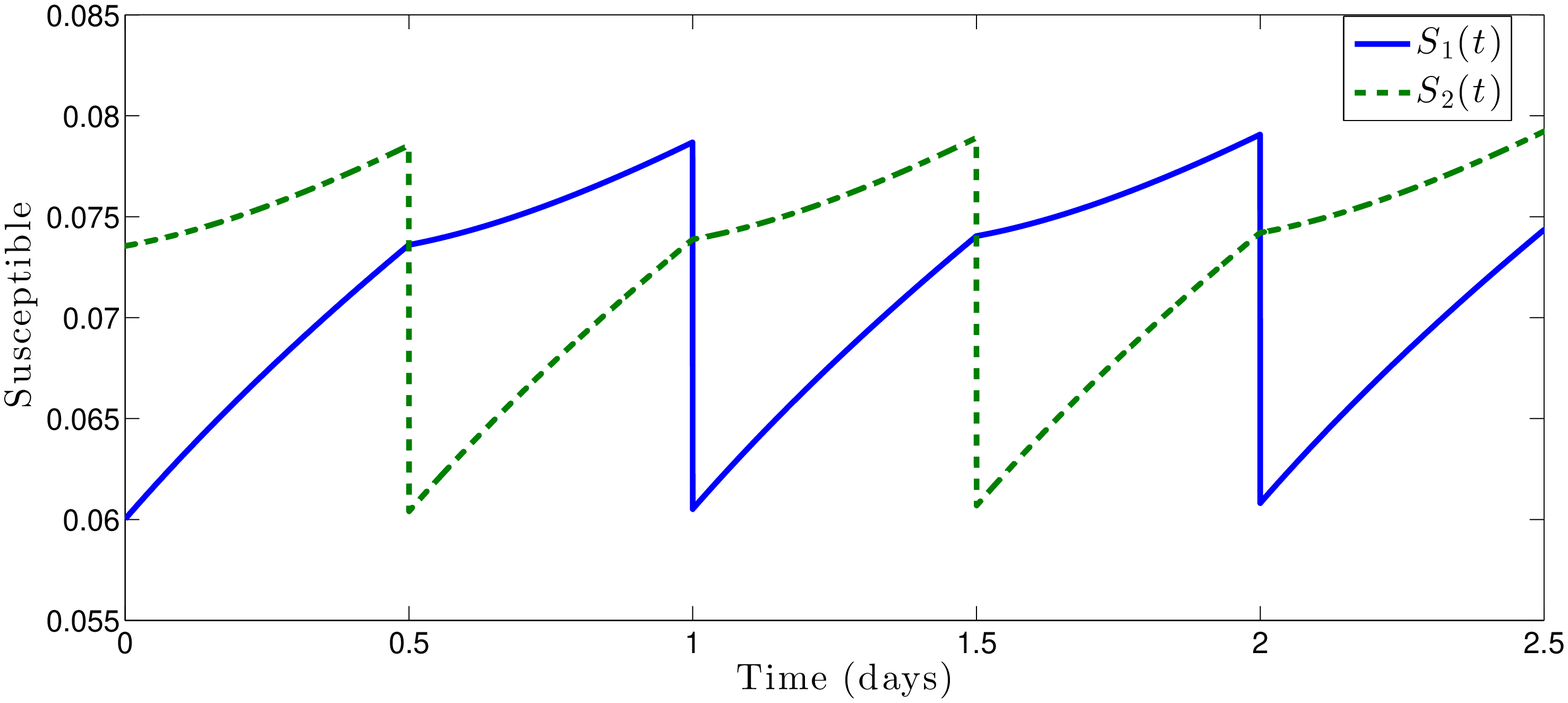}
\caption{The disease-free periodic orbit for the case of two patches with pulses that are administered once a year for each patch out of phase.  Parameters are: $b_1=b_2=\mu_1=\mu_2=1/50$, $\tau_1=\tau_2=\tau=1$, $\psi_1=\psi_2=0.231$, $\phi_1^1=0$ and $\phi_2^1=0.5$.}
  \label{fig:dpu}
  \end{figure}

\section{Reproduction number}\label{sec4}

A threshold between disease eradication and persistence can often be found by utilising the basic reproduction number, $R_0$ \cite{Heffernan}.  This number measures the average number of secondary infections in a wholly susceptible population in the most simple cases or, more generally, the per generation asymptotic growth factor \cite{bacaer2}.  In populations that are not wholly susceptible, such as those that have a significant number of vaccinated individuals, the effective reproductive ratio $R_e$ is used instead.
In many mathematical models, $R_e$ is simply calculated as a local stability threshold. On their own, such local thresholds may not measure the generational asymptotic growth rate and do not account for the possibility of backward bifurcations or other phenomena that may cause the disease to persist when $R_e<1$ \cite{LiBlakeley}.   It is thus crucial that the global dynamics be established and an appropriate definition be instilled for $R_e$ to be meaningful.

The definition of $R_e$ for a general class of periodic population dynamic models was first introduced by Bacaer and Guernaoui in 2006 \cite{bacaer1}.  While a threshold quantity can be often found using Floquet theory, a challenge for defining $R_e$ in periodic non-autonomous models is that the number of secondary cases caused by an infectious individual depends on the season.  The advantage of Bacaer's definition of $R_e$ is that it can be interpreted as the asymptotic ratio of total infections in two successive generations of the infected population and has the threshold properties of the dominant Floquet multiplier.  Wang and Zhao established an equivalent definition of $R_e$ for the case of compartmental periodic ordinary differential equation models \cite{wang}, which we will utilise.  

Define $\overline{S}_i(t), \ i=1,2,\dots,N$ as the components of the $\tau$-periodic orbit $\overline{S}(t)$.  Recall the full periodic model (\ref{mod}).  Define the disease-free $\tau$-periodic orbit, $z(t)$, and the point $\overline{z}\in\mathbb{R}^{3N}_+$ as follows:
\begin{align*}
z(t)&=\left(\overline{S}_1(t),\cdots,\overline{S}_N(t),0,\cdots \cdots,0\right) \\
\overline{z}&=\left(\overline{S}_1(0),\cdots,\overline{S}_N(0),0,\cdots \cdots,0\right).
\end{align*}
Now consider the first variational equation of the ``infectious class'' subsystem along the disease-free periodic orbit $\overline{z}(t)$:
\begin{equation}\label{imod}
\begin{aligned}
\frac{dI_i}{dt}&=\overline{S}_i(t)\sum_j \beta_{ij}(t)I_j+\overline{S}_i(t)\sum_j \epsilon_{ij}(t)G_j-(\mu_i+\gamma_i)I_i  +\sum_{j}k_{ij}I_j  \\
\frac{dG_i}{dt}&=\xi_i(t)I_i-\nu_i(t)G_i 
\end{aligned}
\end{equation}
Subsystem (\ref{imod}) can be written as follows:
\begin{align*}
\frac{dx}{dt}&= (F(t)-V(t))x, 
\end{align*}
where $$x(t)=\left(I_1(t),\dots, I_N(t),G_1(t),\dots,G_N(t) \right)^T \quad \text{and} \quad F_{ij}(t)=\overline{S}_i(t)\left(\beta_{ij}(t)+\epsilon_{ij}(t) \right) \ \text{ for } \ 1\leq i,j\leq N.$$
Here the matrix $F(t)$ represents the ``new infections''.  The matrix $V(t)$ consists of the removal and transition parameters in subsystem (\ref{imod}).

Subsystem (\ref{imod}) is a piecewise-continuous $\tau$-periodic linear differential equation on $\mathbb{R}^{2N}$.  Consider the principal fundamental solution to (\ref{imod}), denoted by $\Phi_{F-V}(t)$.  The Floquet multipliers of the linear system (\ref{imod}) are the eigenvalues of $\Phi_{F-V}(\tau)$.  It can be shown that there is a dominant Floquet multiplier, $r$, which is the spectral radius of $\Phi_{F-V}(\tau)$, i.e.
\begin{align}
r=\rho(\Phi_{F-V}(\tau)). \label{repnum}
\end{align}

Following \cite{wang}, let $Y(t,s)$, $t\geq s$, be the evolution operator of the linear $\tau$-periodic system
\begin{align}
\frac{dy}{dt}&=-V(t)y . \label{vsys}
\end{align}
The principal fundamental solution, $\Phi_{-V}(t)$, of \eqref{vsys} is $Y(t,0)$.  Clearly, $\rho\left(\Phi{-V}(\tau)\right)<1$.  Hence  there exists $K>0$ and $\alpha>0$ such that 
\begin{align*}
\left\| Y(t,s)\right\| \leq Ke^{-\alpha(t-s)}, \quad \forall t\geq s, \quad s\in\mathbb{R}.
\end{align*}
Thus 
\begin{align*}
\left\| Y(t,t-a)F(t-a) \right\| \leq K\left\|F(t-a)\right\|e^{-\alpha a}, \quad \forall t\in\mathbb{R}, \quad a\in [0,\infty).
\end{align*}

Let $\phi(s)$, $\tau$-periodic in $s$, be an initial periodic distribution of infectious individuals.  Given $t\geq s$, $Y(t,s)F(s)\phi(s)$ gives the distribution of those infected individuals who were newly infected at time $s$ and remain in the infected compartments at time $t$.  Then
$$\psi(t)\equiv \int_{-\infty}^t Y(t,s) F(s)\phi(s) \,ds=\int_0^{\infty} Y(t,t-a)F(t-a)\phi(t-a)\,da $$
is the distribution of cumulative new infections at time $t$ produced by all those infected individuals $\phi(s)$ introduced at times earlier than $t$.  

Let $C_{\tau}$ be the ordered Banach space of all $\tau$-periodic piecewise continuous functions from $\mathbb{R}\rightarrow\mathbb{R}^{2N}$, which is equipped with the maximum norm $\left\|\cdot\right\|$.  Define the linear operator $L:C_{\tau}\rightarrow C_{\tau}$ by 
\begin{align*}
(L\phi)(t)=\int_0^{\infty} Y(t,t-a)F(t-a)\phi(t-a)\,da, \quad \forall t\in\mathbb{R}, \ \phi\in C_{\tau}.
\end{align*}
As in \cite{wang}, we label $L$ the next-infection operator and define the spectral radius of $L$ as the effective reproduction number:
\begin{align}
R_e\equiv \rho(L). 
\end{align}
There is a useful characterisation of $R_e$ as follows.  Consider the following linear $\tau$-periodic system
\begin{align}
\frac{dw}{dt}&=\left[-V(t)+\frac{F(t)}{\lambda}\right]w, \quad t\in\mathbb{R} \label{rcomp}
\end{align}
with parameter $\lambda\in (0,\infty)$.  Denote the principal fundamental solution of (\ref{rcomp}) by $\Phi(t,\lambda)$.  Then the following holds
\begin{align}
\rho(\Phi(\tau,\lambda))=1 \Leftrightarrow \lambda=R_e \label{rc}.
\end{align}

Although Wang and Zhao \cite{wang} considered the case where $F(t)$ is continuous, all of the arguments presented in their article apply to the case when $F(t)$ is piecewise continuous, the situation encountered for our system.  In particular,
(\ref{rc}) holds, as do the following equivalences:
\begin{align}
R_e<1& \Leftrightarrow r<1 \label{r01} \\
R_e>1& \Leftrightarrow r>1 \label{r0g1}
\end{align}
where $r=\rho(\Phi_{F-V}(\tau))$.

\section{Threshold dynamics}\label{sec5}

We will show that $R_e$ is a threshold quantity that determines whether the disease dies out or uniformly persists.  First, utilising an asymptotic comparison argument, we prove that the disease-free periodic orbit, $\overline{z}(t)$, is globally attracting for system (\ref{mod}) when $R_e<1$.

\begin{theorem}\label{pgas}
Consider the flow $\varphi(t,x)$ of system (\ref{mod}).  If $R_e<1$ and $x\in\mathbb{R}^{3N}_+$, then $\varphi(t,x)\rightarrow z(t)$ as $t\rightarrow\infty$.  Thus  the disease-free periodic orbit is globally attracting.
\end{theorem}
\begin{proof}
 Let $x\in\mathbb{R}^{3N}_+$.  Consider the solution $\varphi(t,x)=\left(S_1(t),\dots,S_N(t),I_1(t),\dots,I_N(t),G_1(t),\dots,G_N(t)\right)$ of system (\ref{mod}).  By the non-negativity of $I_i(t)$ and $G_i(t)$, we obtain the following:
\begin{align*}
\frac{dS_i}{dt}&\leq (1-p_i)b_i-\mu_iS_i-m_{ii}S_i+\sum_{j\neq i}m_{ij}S_j,  \qquad t\neq n\tau_i+\phi_i \\
S_i\left((n\tau_i+\phi_i)^+\right)&=\left(1-\psi_i\right)S_i\left((n\tau_i+\phi_i)^- \right).
\end{align*}
Define $S(t)=\left(S_1(t),S_2(t),\dots,S_N(t)\right)^T$ and $x(t)$ to be the solution to the disease-free system (\ref{linear}) with $x(0)=S(0)$.  Then, by the above inequality system, we find that $S(t)\leq x(t)$ for all $t\geq 0$.  By Proposition \ref{lineargas}, the $\tau$-periodic solution $\overline{S}(t)$ is globally asymptotically stable for the disease-free system (\ref{linear}).   In particular, $x(t)\rightarrow\overline{S}(t)$ as $t\rightarrow\infty$.  Fix $\epsilon>0$.  Then there exists $t_1>0$ such that $x_i(t)\leq\overline{S}_i(t)+\epsilon \ \ \forall t\geq t_1$.  Thus  $S_i(t)\leq \overline{S}_i(t)+\epsilon$ for all $t\geq t_1$.  Hence, for all $t\geq t_1$,
\begin{equation}\label{rhs}
\begin{aligned}
\frac{dI_i}{dt}&\leq (\overline{S}_i(t)+\epsilon)\sum_j \beta_{ij}(t)I_j+(\overline{S}_i(t)+\epsilon)\sum_j \epsilon_{ij}(t)G_j-(\mu_i+\gamma_i)I_i  -k_{ii}I_i+\sum_{j\neq i}k_{ij}I_j  \\
\frac{dG_i}{dt}&=\xi_i(t)I_i-\nu_i(t)G_i. 
\end{aligned}
\end{equation}

Consider the principal fundamental solution of the right-hand side of system \eqref{rhs} as a function of $\epsilon$:  $\Phi(t,\epsilon)$.  Then $\rho(\Phi(\tau,0))=r$.  The periodic solution $\overline{S}(t)$ is piecewise continuous with a total of $p$ points of discontinuity on the interval $(0,\tau]$.  Thus  the same is true for the right-hand side of (\ref{rhs}).  On each piece, solutions have continuous dependence on parameters.  Therefore we can conclude that solutions will have continuous dependence on parameters for the whole interval $[0,\tau)$.  Hence $\Phi(\tau,\epsilon)$ is continuous with respect to $\epsilon$.  So, for $\epsilon$ sufficiently small, $r(\epsilon)=\rho(\Phi(\tau,\epsilon))<1$ since $r(0)=r<1$ by (\ref{r01}).  The matrix $B(t,\epsilon)$, where $B(t,\epsilon)$ represents the right-hand side of (\ref{rhs}) as a linear vector field, is quasi-positive.  Without loss of generality, we can assume the non-diagonal entries of $B(t,\epsilon)$ are positive.  If any are zero, add a sufficiently small constant to that entry and the spectral radius of interest will still fall below unity, and inequality (\ref{rhs}) will still hold.  Thus  the matrix $\Phi(\tau,\epsilon)$ will be strictly positive (since the vector field will point away from the boundary).  Then, by the Perron--Frobenius Theorem, we find that $r(\epsilon)$ is a simple eigenvalue with strictly positive eigenvector $v$.  Hence $y(t)\equiv \Phi(t,\epsilon)v=q(t)e^{\alpha t}$ where $\alpha=\frac{1}{\tau}\ln(r(\epsilon))$ and $q(t)$ is $\tau$-periodic.  So $y(t)\rightarrow 0$ as $t\rightarrow\infty$.  Since $B(t,\epsilon)$ is quasi-positive, subsystem (\ref{rhs}) forms a comparison system using Theorem 1.2 in \cite{comparison}.  Choose a constant $c$ such that $cv\geq x_1$, where $x_1=\varphi(t_1,x)$.  Then $cy(t)\geq \varphi(t,x_1)$ for all $t\geq 0$.  Hence $I_i(t)\rightarrow 0$ and $G_i(t)\rightarrow 0$ as $t\rightarrow\infty$ for $i=1,\dots,N$ since $cy(t)\rightarrow 0$.  Then, for any $\epsilon>0$ and sufficiently large time $t$,
\begin{align*}
(1-p_i)b_i-\mu_iS_i-m_{ii}S_i+\sum_{j\neq i}m_{ij}S_j-\epsilon\leq\frac{dS_i}{dt}&\leq (1-p_i)b_i-\mu_iS_i-m_{ii}S_i+\sum_{j\neq i}m_{ij}S_j,  \qquad t\neq n\tau_i+\phi_i  \\
S_i\left((n\tau_i+\phi_i)^+\right)&=\left(1-\psi_i\right)S_i\left((n\tau_i+\phi_i)^-\right)
\end{align*}

The impulsive system representation of the left-hand side of the above inequality also has a globally stable $\tau$-periodic solution $\overline{S}^{\epsilon}(t)=\left(\overline{S}_1^{\epsilon}(t),\dots,\overline{S}_N^{\epsilon}(t)\right)$.  Another application of the comparison system principle yields $\overline{S}_i^{\epsilon}(t)\leq S_i(t)\leq \overline{S}_i(t)$ for $t$ sufficiently large.  Continuous dependence on parameters implies that $\overline{S}_i^{\epsilon}(t)$ can be made arbitrarily close to $\overline{S}_i(t)$ as $\epsilon\rightarrow 0$.  Clearly, the fixed point equation $F(x)=x$ (from the proof of (\ref{lineargas})) depends continuously on the matrix $A$.  Thus $\overline{S}^{\epsilon}_i(0)>\overline{S}(0)-\epsilon_1$ where $\epsilon_1$ is arbitrary and $\epsilon$ is chosen sufficiently small.  The functions $\overline{S}^{\epsilon}(t)$ and $\overline{S}(t)$ are uniformly continuous for $t\neq t_{\ell}$ where $1\leq \ell\leq p$.   It follows that if $\epsilon$ is chosen sufficiently small, then $\overline{S}^{\epsilon}(t)>\overline{S}(t)-\epsilon_2$ for all $t\in[0,\tau)$ where $\epsilon_2$ is arbitrary. The result follows.
\end{proof}

We now turn our attention the dynamics when $R_e>1$.  In order to prove that the disease is uniformly persistent in all patches when $R_e>1$, we need to make extra assumptions on the $N\times N$ $\tau$-periodic matrix $M(t)\equiv\left(\beta_{ij}(t)+k_{ij}+\xi_i(t)\epsilon_{ij}(t)\right)_{1\leq i,j\leq N}$.  Assume that: 
\begin{itemize}
\item[(A1)] There exists $\theta \in [0,\tau)$ such that $M(\theta)$ is irreducible.
\end{itemize}
Biologically, this irreducibility assumption means that, at some time during a period, the patches have the property that infection in an arbitrary patch can cause infection in any other patch through some chain of transmissions or migrations among a subset of patches.  If this assumption is satisfied, then the system is uniformly persistent, detailed in the following theorem.
\begin{theorem} \label{persist}
Suppose that $R_e>1$ and (A1) holds.  Then the system (\ref{mod}) is uniformly persistent; i.e., there exists $\delta>0$ such that
if $\beta_{ij}I_j(0)>0$ or $\epsilon_{ij}G_j(0)>0$, for some $1\leq i,j\leq N$,  then  
$$ \liminf_{t\rightarrow\infty}I_i(t) > \delta \quad \forall i=1,\dots, N.$$
\end{theorem}

\begin{proof}
We intend to use the approach of acyclic coverings to prove uniform persistence.  We will use Theorem 1.3.1 from \cite{zhao}.   Let $X\equiv \mathbb{R}^{3N}_+$, $ X^0\equiv \left\{x\in X: x_i>0 \ N+1\leq i\leq 3N \right\}$ and $\partial X^0=X \setminus X^0$. Define the Poincar\'{e}map $P(x)=\varphi(\tau,x)$, where $\varphi(t,x)$ is a solution to the full system (\ref{mod}).  Note that $P: X\rightarrow X$ is a continuous map on the complete metric space $X$.  In addition, $X^0$ is forward invariant under the semiflow $\varphi(t,x)$ and hence $P(X^0)\subset X^0$.  Define the maximal forward invariant set inside $\partial X^0$ by  $M_{\partial}\equiv \left\{x\in\partial X^0:  P^n(x)\in\partial X^0 \ \forall n\in\mathbb{N} \right\}$.  
First, we show that $P$ is uniformly persistent; i.e., there exists $\epsilon>0$ such that, for all $x\in X^0$, $\liminf_{n\rightarrow\infty} d(P^n(x),\partial X^0)>\epsilon$.
Note that $P$ is a compact map and is point dissipative by Proposition \ref{P1}.  The global attractor of $P$ in $M_{\partial}$ is the singleton $\left\{\overline{z}\right\}$ by Proposition \ref{lineargas}.  Therefore 
$$\bigcup_{x\in M_{\partial} } \omega(x)=\left\{\overline{z}\right\} .$$ 

On the boundary subset $M_{\partial}$, $P(x)=\left(F(x_S),0,\dots,0,\dots,0\right)$ where $x_S=\left(x_1,\dots,x_N\right)$ and $F$ is defined in Proposition \ref{lineargas}.  Let $x\in\partial X^0\setminus\left\{\overline{z}\right\}$.  Then $\left| P^{-n}x-\overline{z}\right|=\left|F^{-n}(x_S)-F^{-n}(\overline{x})\right|=\left|C^{-n}(x_S-\overline{x})\right|\rightarrow\infty$ since all eigenvalues of $C^{-1}$ are greater than unity (where $\overline{x}$ and $C$ are defined in Proposition \ref{lineargas}).  Thus $\left\{\overline{z}\right\}$ is acyclic.  

We next show that $\left\{\overline{z}\right\}$ is isolated.  Consider the derivative of the Poincar\'{e} map evaluated at $\overline{z}$, $DP(\overline{z})$.  Note that the eigenvalues of $DP(\overline{z})$ are also the Floquet multipliers of the linearized system (\ref{mod}) along the disease-free periodic orbit $z(t)$.  The linearization matrix is block triangular.  It can be seen that $\rho(DP(\overline{z}))=r>1$.  By assumption (A1), the eigenvector corresponding to $r$, which we call $u$, has positive ``infection components''; i.e., $u_i>0$, $N+1\leq i\leq 3N$.  An application of the stable manifold theorem for discrete-time dynamical systems implies that $\left\{\overline{z}\right\}$ is isolated.

Therefore the remaining hypothesis to check is that $W^s(\left\{\overline{z}\right\})\cap X^0 =\emptyset$.   By way of contradiction, suppose that there exists $x\in X^0$ such that $P^n(x)\rightarrow \overline{z}$ as $n\rightarrow\infty$.  Let $\epsilon>0$ be arbitrary.  Then there exists $N(\epsilon_0)>0$ such that $\left| P^n(x) -\overline{z} \right|<\epsilon_0 \ \forall n\geq N(\epsilon_0)$.  In particular, $I_i(n\tau), G_i(n\tau)<\epsilon_0$ for all $n\geq N(\epsilon_0)$.  Notice that the functions $I_i(t)$ and $G_i(t)$ for $t\geq 0$  ($i=1,\dots, N$) are uniformly continuous since their derivatives are bounded for all $t\geq 0$.  By this uniform continuity and the compactness of $[n\tau,(n+1)\tau]$, for any $\epsilon_1>0$, we can choose $\epsilon_0$ sufficiently small so that $I_i(t), G_i(t)<\epsilon_1$ for all $t\geq N(\epsilon_0)\tau$.  Then
\begin{equation}\label{dfeps}
\begin{aligned}
\frac{dS}{dt}&\leq b+(A-\epsilon_2 I )S, \qquad t\neq n\tau+t_{\ell}  \\
x\left(n\tau+t_{\ell}\right)&=D_{\ell} x\left((n\tau+t_{\ell})^-\right), 
\end{aligned}
\end{equation}
where $n\geq N(\epsilon_0)$, $S=\left(S_1,\dots, S_N\right)^T$, $I$ is the $N\times N$ identity matrix, $\epsilon_2=2\epsilon_1 \max_{i,j}\left(\beta_{ij},\epsilon_{ij}\right)$ and $A, b, t_{\ell}, D_{\ell}$ are defined in (\ref{linear}).  Thus, by the standard comparison theorem and (\ref{lineargas}),
\begin{align*}
S(t)&\geq S^{\epsilon_2}(t)>\overline{S}^{\epsilon_2}(t)-\epsilon_3 \quad \forall t\geq \max(T(\epsilon_3),N(\epsilon_0)),
\end{align*}
where $S^{\epsilon_2}(t)$ is the solution to the linear impulsive equation forming the right-hand side of (\ref{dfeps}) and $\overline{S}^{\epsilon_2}(t)$ is the globally stable periodic impulsive orbit in this system.  By continuous dependence on parameters, if $\epsilon_2$ is chosen sufficiently small, then $\overline{S}^{\epsilon_2}(t)>\overline{S}(t)-\epsilon_4$ for all $t\in[0,\tau)$ where $\epsilon_4$ is arbitrary (the argument is presented in the proof of Theorem \ref{pgas}).   Therefore 
\begin{align*}
S(t)&>\overline{S}(t)-\epsilon_4-\epsilon_3 \qquad \text{for} \ t>\max(T(\epsilon_3),N(\epsilon_0)\tau).
\end{align*}
Hence, for $\epsilon_0$ sufficiently small, there exists $N(\epsilon_0)\in\mathbb{N}$ such that 
\begin{align*}
S(t)&>\overline{S}(t)-\frac{\epsilon}{2}-\frac{\epsilon}{2} \qquad \text{for} \ t>\max(T(\epsilon/2),N(\epsilon_0)\tau).
\end{align*}
By the semigroup property and the fact that $X^0$ is forward invariant,  we can assume without loss of generality that
\begin{align*}
S(t)&>\overline{S}(t)-\epsilon \qquad \text{for} \ t\geq 0.
\end{align*}
Then
\begin{equation}\label{rhs2}
\begin{aligned}
\frac{dI_i}{dt}&\geq (\overline{S}_i(t)-\epsilon)\sum_j \beta_{ij}(t)I_j+(\overline{S}_i(t)-\epsilon)\sum_j \epsilon_{ij}(t)G_j-(\mu_i+\gamma_i)I_i  -k_{ii}I_i+\sum_{j\neq i}k_{ij}I_j  \\
\frac{dG_i}{dt}&=\xi_i(t)I_i-\nu_i(t)G_i. 
\end{aligned}
\end{equation}

By the comparison theorem, $y(t)\geq \widetilde{y}(t)$, where $y(t)=(I_1,\dots,I_N(t),G_1(t)\dots,G_N(t))$ and $\widetilde{y}(t)$ is a vector solution to the right-hand side of \eqref{rhs2} with $\widetilde{y}(0)\leq y(0)$.  For $\epsilon>0$ sufficiently small, $r(\epsilon)>1$ (by (\ref{r0g1})) where $r(\epsilon)$ is the dominant Floquet multiplier of the right-hand side of (\ref{rhs2}).  By the Perron--Frobenius Theorem, we find that $r(\epsilon)$ is a simple eigenvalue with strictly positive eigenvector $v$.  Hence there is a vector solution $\widetilde{y}(t)\equiv \Phi(t,\epsilon)v=q(t)e^{\alpha t}$ where $\alpha=\frac{1}{\tau}\ln(r(\epsilon))$ and $q(t)$ is $\tau$-periodic.  Then $c\widetilde{y}(t)$ is also a solution and, for $c>0$ sufficiently small, $c\widetilde{y}(0)<y(0)$.  Notice that $c\widetilde{y}(n\tau)=c\left(r(\epsilon)\right)^n$.  Thus  $\left|c\widetilde{y}(n\tau)\right|\rightarrow\infty$ as $n\rightarrow\infty$ and therefore $\left|y(n\tau)\right|\rightarrow \infty$ as $n\rightarrow\infty$.  This is a contradiction.  This proves that $P$ is uniformly persistent; i.e., 
there exists $\epsilon>0$ such that, for all $x\in X^0$, $\liminf_{n\rightarrow\infty} d(P^n(x),\partial X^0)>\epsilon$.

The next step is to prove that there exists $\delta>0$ such that, for all $x\in X^0$, $\liminf_{t\rightarrow\infty} d(\varphi(t,x),\partial X^0)>\delta$.
From an argument presented in the proof of Proposition \ref{lineargas},  for any solution to the impulsive model (\ref{mod}), $S_i(t)\leq \overline{S}_i(t) +1$ for all $t$ sufficiently large. By way of contradiction, suppose that there exists $\delta_m\downarrow 0$ and $(x_m)\subset X^0$ such that 
$$ \liminf_{t\rightarrow\infty} d(\varphi(t,x_m),\partial X^0)\leq \delta_m .$$
For sufficiently large $t$, any solution satisfies
\begin{equation}\label{rhs3}
\begin{aligned}
\frac{dI_i}{dt}&\leq (\overline{S}_i(t)+1)\sum_j \beta_{ij}(t)I_j+(\overline{S}_i(t)+1)\sum_j \epsilon_{ij}(t)G_j-(\mu_i+\gamma_i)I_i  -k_{ii}I_i+\sum_{j\neq i}k_{ij}I_j  \\
\frac{dG_i}{dt}&=\xi_i(t)I_i-\nu_i(t)G_i. 
\end{aligned}
\end{equation}

By Floquet's theorem \cite{chicone}, the principal fundamental solution of the right-hand side of (\ref{rhs3}) can be represented as $\Phi_1(t)=Q(t)e^{tB}$ where $Q(t)$ is a $\tau$-periodic (possibly complex) matrix and $B$ is a (possibly complex) $2N\times 2N$ matrix.  Let $K=\left\|\Phi_1(\tau)\right\|\leq e^{\tau\left\|B\right\|}<\infty$.  Choose $\delta_m>0$ such that $\delta_m<\frac{\epsilon}{K}$.  Then, for some sufficiently large $n$, $\varphi((n+1)\tau,x_m)> \epsilon$ and $\varphi(n\tau+t^*,x_m)\leq \delta_m$, where $t^*\in(0,\tau)$.   Let $I_i(t)=\varphi_{N+i}(n\tau+t^*+t,x_m)$, $G_i(t)=\varphi_{2N+i}(n\tau+t^*+t,x_m)$ for $i=1,\dots,N$ and $y(t)=(I_1,\dots,I_N(t),G_1(t)\dots,G_N(t))$.  Then, by the comparison theorem, $y(\tau-t^*)\leq \Phi_1(n-t^*)\delta_m$.  Thus $\left|y(\tau-t^*)\right|\leq \left\|\Phi_1(\tau-t^*)\right\| |x_m| \leq \left\|\Phi_1(\tau)\right\| \delta_m< K\cdot \frac{\epsilon}{K} <\epsilon$.  Equivalently, $\varphi((n+1)\tau,x_m)< \epsilon$, which is a contradiction.  Therefore there exists $\delta>0$ such that, for all $x\in X^0$, $\liminf_{t\rightarrow\infty} d(\varphi(t,x),\partial X^0)>\delta$,
which proves the result.
\end{proof}

\section{Two-patch case with application to Poliomyelitis eradication}\label{sec6}
In order to show how our model and analysis can inform the optimal timing of pulse vaccination, along with the effect of other parameters, we consider an example of two coupled patches.  For simplicity, we initially neglect environmental transmission and consider the following special case of the model (\ref{mod}):
\begin{equation}\label{2patch}
\begin{aligned}
\dfrac{dS_1}{dt}&=b_1-\mu_1 S_1-\beta_1(t) S_1\left[(1-f_1)I_1+f_1\cdot I_2\right]-m_1S_1+m_2S_2 & t&\neq n\tau \\ 
\dfrac{dI_1}{dt}&=\beta_1(t) S_1\left[(1-f_1)I_1+f _1\cdot I_2\right]-(\mu_1+\gamma_1)I_1-k_1I_1+k_2I_2  & t&\neq n\tau  \\
\dfrac{dS_2}{dt}&=b_2-\mu_2 S_2- \beta_2(t) S_2\left[f_2\cdot I_1+(1-f_2)I_2\right]-m_1S_1+m_2S_2 &    t&\neq n\tau+\phi \\ 
\dfrac{dI_2}{dt}&=\beta_2(t) S_2\left[f_2\cdot I_1+(1-f_2)I_2\right]-(\mu_2+\gamma_2)I_2-k_1I_1+k_2I_2  &    t&\neq n\tau+\phi \\
S_1\left(n\tau^+\right)&=(1-\psi_1) S_1\left(n\tau^-\right) & t &= n\tau \\
 S_2\left((n\tau+\phi)^+\right)&=(1-\psi_2 )S_2\left((n\tau+\phi)^-\right)& t &= n\tau+\phi, 
\end{aligned}
\end{equation}
where $n\in \mathbb{N}$, $0\leq \phi < 1$ is the phase difference between pulse vaccinations in each patch and $0\leq f_1, f_2 \leq \frac{1}{2}$ is the mass-action coupling factor (fraction of cross transmission) in Patch 1 and Patch 2, respectively.  In the following subsections, we provide deeper analysis and simulations of the model in the case of two patches.  

\subsection{Pulse synchronisation theorem}
A natural question to ask is how does the relative timing of the pulse vaccinations in the two individual patches affect the global dynamics of the disease.  By Theorems \ref{pgas} and \ref{persist}, the effective reproduction number $R_e$ provides a global threshold.  Thus it suffices to determine how the phase difference between pulse vaccinations in the two patches affects the value of $R_e$.  In general, $R_e$ cannot be calculated explicitly.  However, in the special case of no cross-infection or movement of infected individuals --- i.e., $f_1=f_2=0$ and $k_1=k_2=0$ --- $R_e$ can be formulated; furthermore, the effect of $\phi$ on $R_e$ can be quantified.  In particular, we can prove that pulse synchronisation ($\phi=0$) minimizes (locally) the reproduction number $R_e$.  The details are contained in the following theorem.

\begin{theorem} \label{R0theorem}
Consider the two-patch model (\ref{2patch}) with phase difference $\phi\in\mathbb R$ between pulse vaccinations and no cross-infection or movement of infected individuals; i.e., $f_1=f_2=0$ and $k_1=k_2=0$.  Let $\overline{S}_1^{\phi}(t)$ and $\overline{S}_2^{\phi}(t)$ be the disease-free periodic solutions given by Proposition \ref{lineargas} parametrized by the phase difference $\phi$.   Then
\begin{align} R_e(\phi)= \max_{i=1,2}\left\{\frac{1}{(\mu_i+\gamma_i)\tau}\int_0^{\tau} \beta_i(t)\overline{S}_i^{\phi}(t)\,dt\right\}, \label{R0diag}
\end{align}
where the reproduction number is a $\tau$-periodic function of $\phi$ on $\mathbb R$.  
Suppose further that the patches are identical and transmission rates are constant; i.e., $b_1=b_2=b, \mu_1=\mu_2=\mu, \beta_1(t)=\beta_2(t)=\beta, \gamma_1=\gamma_2=\gamma, m_1=m_2=m$ and $\psi_1=\psi_2=\psi$.  Then $R_e(\phi)$ has a local minimum at $\phi=0$ (in-phase pulses) and a critical point at $\phi=\tau/2$ (out-of-phase pulses).  More precisely, $R_e(\phi)$ is continuous on $\mathbb R$, smooth on $\mathbb R\setminus \left\{n\tau:n\in\mathbb Z \right\}$ and 
\begin{align*}
R_e'(0^+)& = \widehat{R}_0\frac{m \psi^2(e^{\mu\tau}-1)^2}{2\tau\left(e^{\mu\tau}-(1-\psi)\right)^2}>0, \qquad
R_e'(0^-) \ \Big(=R_e'(\tau^-)\Big) \ =-R_e'(0^+)<0, \qquad R_e'\left(\frac{\tau}{2}\right)=0,
\end{align*}
where $\widehat{R}_0=\frac{\beta b}{\mu(\mu+\gamma)}$ is the reproduction number of the (identical patch) system in the absence of pulse vaccination.
\end{theorem}

\begin{proof}
Consider the two-patch model (\ref{2patch}).   Define the following linear system as in (\ref{rcomp}):
\begin{align}
\frac{dw}{dt}&=\left[-V+\frac{F_{\phi}(t)}{\lambda}\right]w, \quad t\in\mathbb{R}, \label{rphi}
\end{align}
where 
$$V=\begin{pmatrix} \mu_1+\gamma_1+k_1 & -k_2 \\ -k_1 & \mu_2+\gamma_2+k_2 \end{pmatrix}, \qquad F_{\phi}(t)=\begin{pmatrix} \beta_1(t) \overline{S}_1^{\phi}(t)(1-f_1) & \beta_1(t) \overline{S}_1^{\phi}(t)f_1 \\ \beta_2(t) \overline{S}_2^{\phi}(t)f_2 & \beta_2(t) \overline{S}_2^{\phi}(t)(1-f_2) \end{pmatrix}.  $$
Define the principle fundamental solution of (\ref{rphi}) as $W_{\phi}(t,\lambda)$.  Then the reproduction number as a function of $\phi$, $R_e(\phi)$, is defined as the unique value of $\lambda$ such that the dominant Floquet multiplier of (\ref{rphi}) is one; i.e., $\rho(W_{\phi}(\tau,\lambda))=1$.  Clearly $R_e(\phi)$ is a $\tau$-periodic function.  If $k_1=k_2=f_1=f_2=0$, then the matrix $B_{\phi}(t,\lambda)\equiv-V+\frac{F_{\phi}(t)}{\lambda}$ is diagonal.  Thus, as noted in \cite{wang} for diagonal systems, the eigenvalues of $W_{\phi}(\tau,\lambda)$ are $r_i=-(\mu_i+\gamma_i)\tau+\frac{1}{\lambda}\int_0^{\tau} \beta_i(t)\overline{S}_i^{\phi}(t)\,dt$, for $i=1,2$.  It follows that 
$$ R_e(\phi)= \max_{i=1,2}\Big(\frac{1}{(\mu_i+\gamma_i)\tau}\int_0^{\tau} \beta_i(t)\overline{S}_i^{\phi}(t)\,dt\Big).$$
In order to compute the derivative of $R_e(\phi)$, let $\phi>0$ and define $R'(\phi)$, and $R'(\phi^+)$ and $R_e'(\phi^-)$ as the derivative and one-sided derivatives of $R_e(\phi)$ respectively.  Consider the eigenvalue $r(\lambda,\phi)=\rho(W_{\phi}(\tau,\lambda))$.  The characteristic equation is 
$$r(\lambda,\phi)^2-{\rm tr}W_{\phi}(\tau,\lambda)r(\lambda,\phi)+\det W_{\phi}(\tau,\lambda)=0.$$
Then, since $r(R_e(\phi),\phi)=1$ for all $\phi$, we obtain:
\begin{align}
\det W_{\phi}(\tau,R_e(\phi))&={\rm tr}W_{\phi}(\tau,R_e(\phi))-1 \notag \\
\Leftrightarrow {\rm exp}\left(\int_0^{\tau}{\rm tr}B_{\phi}(t,R_e(\phi)) \,dt\right)&= {\rm tr}W_{\phi}(\tau,R_e(\phi))-1 \qquad \text{(by Liouville's formula)} \notag \\
\text{Thus} \ \  \frac{\partial}{\partial \phi} {\rm exp}\left(\int_0^{\tau}{\rm tr}B_{\phi}(t,R_e(\phi)) \,dt\right)&= \frac{\partial}{\partial \phi} {\rm tr}W_{\phi}(\tau,R_e(\phi)). \notag 
\end{align}
Calculating the derivative with respect to $\phi$, we obtain:
\begin{align}
{\rm exp}\left(\int_0^{\tau}{\rm tr}B_{\phi}(t,R_e(\phi)) \,dt\right) \frac{1}{R_e(\phi)} \left[  \frac{\partial}{\partial \phi} \int_0^{\tau}{\rm tr}F_{\phi}(t) \,dt   -\frac{R_e'(\phi)}{R_e(\phi)} \int_0^{\tau}{\rm tr}F_{\phi}(t) \,dt \right] = \frac{\partial}{\partial \phi} {\rm tr}W_{\phi}(\tau,R_e(\phi)). \label{R0prime} 
\end{align}
In the case that $B_{\phi}(t,R_e(\phi))$ is diagonal --- i.e., $k_1=k_2=f_1=f_2=0$ --- then $W_{\phi}(\tau,R_e(\phi)) ={\rm exp}\left(\int_0^{\tau}B_{\phi}(t,R_e(\phi)) \,dt\right)$.   Therefore
\begin{align*}
 \frac{\partial}{\partial \phi} {\rm tr}W_{\phi}(\tau,R_e(\phi)) &=\sum_{i=1}^2{\frac{e^{-(\mu_i+\gamma_i)\tau}}{R_e(\phi)}\left[  \frac{\partial}{\partial \phi} \int_0^{\tau}\beta_i(t)\overline{S}_i^{\phi}(t) \,dt   -\frac{R_e'(\phi)}{R_e(\phi)}  \int_0^{\tau}\beta_i(t)\overline{S}_i^{\phi}(t) \,dt   \right]. } 
\end{align*}
At this stage, we assume that $\mu_1=\mu_2=\mu$ and $\gamma_1=\gamma_2=\gamma$.  Inserting the above equation into (\ref{R0prime}), simplifying, and solving for $R_e'(\phi)$, we obtain:  
\begin{align}
R_e'(\phi)&= R_e(\phi)\frac{ \frac{\partial}{\partial \phi}\left[ \int_0^{\tau}\left(\beta_1(t)\overline{S}_1^{\phi}(t)+\beta_2(t)\overline{S}_2^{\phi}(t)\right) \,dt\right]}{  \int_0^{\tau}\left(\beta_1(t)\overline{S}_1^{\phi}(t)+\beta_2(t)\overline{S}_2^{\phi}(t)\right) \,dt}. \label{R0pElegant}  
\end{align}
The above formula shows that, under the prescribed assumptions, the relative change in $R_e$ with respect to the parameter $\phi$ is equal to the relative change (with respect to $\phi$) in total force of infection summed over both patches and averaged over the period $\tau$.  In the case that $\beta_1(t)=\beta_2(t)=\beta$ is constant, $\beta$ cancels in formula (\ref{R0pElegant}) and we only need to consider the effect of $\phi$ on the total susceptible population among the patches averaged over the period.  

Adding the differential equations of the disease-free $\tau$ periodic solutions and integrating over the period $\tau$, we have
\begin{align*}
0= \int_0^{\tau}\frac{d}{dt}\left(\overline{S}_1^{\phi}(t)+\overline{S}_2^{\phi}(t)\right) \,dt&=(b_1+b_2)\tau -\mu\int_0^{\tau}\left(\overline{S}_1^{\phi}(t)+\overline{S}_2^{\phi}(t)\right) \,dt - \int_0^{\tau}\left(\psi_1\delta_0\overline{S}_1^{\phi}(t)+\psi_2\delta_{\phi}\overline{S}_2^{\phi}(t)\right) \,dt,  
\end{align*}
where $\delta_{\phi}$ is the Dirac delta mass centred at $\phi$.  Thus
\begin{align}
\int_0^{\tau}\left(\overline{S}_1^{\phi}(t)+\overline{S}_2^{\phi}(t)\right) \,dt &=\frac{1}{\mu}\left[(b_1+b_2)\tau-\left(\psi_1\overline{S}_1^{\phi}(0^-)+\psi_2\overline{S}_2^{\phi}(\phi^-)\right)\right]. \label{pulse}
\end{align}
Let
$$A=\begin{pmatrix} -\mu-m_1 & m_2 \\ m_1 & -\mu-m_2 \end{pmatrix}. $$
As in the earlier formula (\ref{F1}), define the initial point of the disease-free periodic solution $\overline{x}=\left(\overline{S}_1^{\phi}(0^-), \overline{S}_2^{\phi}(0^-)\right)^T$, along with the matrices $D_1={\rm diag}(1-\psi_1,1)$, $D_2={\rm diag}(1,1-\psi_2)$, $C_2=e^{(\tau-\phi)A}D_2$, $C=C_2e^{\phi A}D_1$.  Also define $\overline{x}_{\phi}=\left(\overline{S}_1^{\phi}(\phi^-), \overline{S}_2^{\phi}(\phi^-)\right)^T$.  Then
$$\overline{x}=(I-C)^{-1}\left(C_2(e^{\phi A}-I)+e^{(\tau-\phi)A}-I\right)A^{-1}b, \qquad \overline{x}_{\phi}=e^{\phi A}D_1\overline{x}+\left(e^{\phi A}-I\right)A^{-1}b. $$
All times are considered to be ${\rm modulo} \ \tau$.  It can be inferred from these equations, along with (\ref{pulse}) and (\ref{R0pElegant}), that $R_e'(\phi)$ is continuous on the set $[0,\tau)$ (where the interval is identified topologically as a circle glued together at the endpoints) and continuously differentiable on the set $(0,\tau)$.  We further assume that $m=m_1=m_2$, $\psi=\psi_1=\psi_2$.  Inserting the above formulas into (\ref{pulse}), utilising  (\ref{R0pElegant}), and evaluating $\partial/\partial\phi$ at $\phi=0^+$, $\phi=0^-$ and $\phi=\tau/2$, we find the following information:
\begin{align}
R_e'(0^+)&=\frac{\psi^2mb\left(e^{\mu\tau}-1\right)^2}{\mu\left(1-\psi-e^{\mu\tau}\right)^2} \cdot \frac{R_e(\phi)}{ \int_0^{\tau}\left(\overline{S}_1^{0}(t)+\overline{S}_2^{0}(t)\right) \,dt} = \widehat{R}_0\frac{m \psi^2(e^{\mu\tau}-1)^2}{2\tau\left(e^{\mu\tau}-(1-\psi)\right)^2}>0 \\
R_e'(0^-)&=R_e'(\tau^-)=\frac{-\psi^2mb\left(e^{\mu\tau}-1\right)^2}{\mu\left(1-\psi-e^{\mu\tau}\right)^2} \cdot \frac{R_e(\phi)}{ \int_0^{\tau}\left(\overline{S}_1^{\phi}(t)+\overline{S}_2^{\phi}(t)\right) \,dt}=-R_e'(0^+)<0  \\
R_e'\left(\frac{\tau}{2}\right)&=0.
\end{align}
\end{proof}

\noindent
{\bf Remarks.} 1. The explicit formula for $R_e'(0^+)$ can tell us which parameters affect the sensitivity of $R_e$ to $\phi$.  In particular, $R_e'(0^+)$ is increasing with respect to the migration rate $m$, the ``natural'' effective reproduction number $\widehat{R}_0$, the death rate $\mu$ and the pulse vaccination proportion $\psi$.  Thus an increase in any of these parameters results in larger increases in $R_e$ when the pulse vaccinations in the two patches are perturbed away from synchrony.  \\
2. While it would be nice to obtain more general results, the mathematical complexity of the system is difficult to overcome.  We suspect that $\phi=0$ is the global minimum, but the formulas for $R_e(\phi)$ and $R_e'(\phi)$ are difficult to analyze at other values of $\phi$.  We note that simulations show that, for small enough values of migration rate $m$, $\phi=\tau/2$ corresponds to a global maximum of $R_e(\phi)$.  However, for large migration rates $m$, $\phi=\tau/2$ may not correspond to a maximum, and simulations show there can be a value $0<\phi^*<\tau/2$ such that $\phi=\phi^*$ and $\phi=\tau-\phi^*$, which are both global maxima of $R_e(\phi)$. See Figure \ref{LargeMigEx}. \\
3. It would be nice to remove the assumption of no cross-infection or migration of infected individuals. However, the linear periodic system cannot generally be solved in this case, since $B_{\phi}(t,\lambda)\equiv-V+\frac{F_{\phi}(t)}{\lambda}$ is not diagonal.  The system can be explicitly solved in this case when $\phi=0$ and the patches are identical (since $B_{0}(t,\lambda)$ commutes with its integral in this case), but a perturbation in $\phi$ is very difficult to analyze, and we leave this as future work.

\begin{figure}
\begin{centering}
\includegraphics[width=9cm]{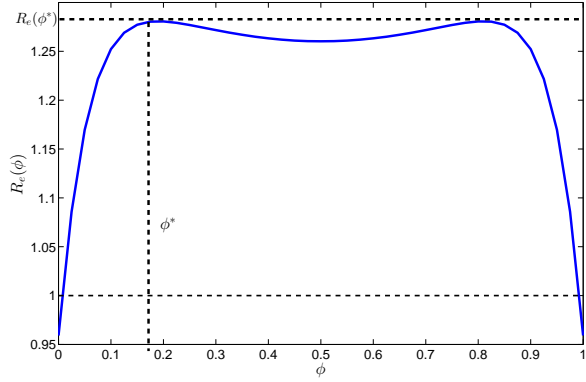}
\caption{Example of large migration rates producing two global maxima of $R_e(\phi)$. Parameters used were $m=10$, $\widehat{R}_0=20$, $\mu=b=1/50$, $\gamma=365/16$, and $\psi=0.8$.}\label{LargeMigEx}
\end{centering}
\end{figure}

\bigskip

\noindent
In the next subsection, we provide simulations that illustrate the theorem and show that the implications hold in the case of cross-infection and migration of infected individuals.  In addition, the pulse synchronisation result is explained in the context of mass-vaccination campaigns against polio and measles. 

\subsection{The SIR model with identical patches and no seasonality}\label{sec61}
For the simulations in this subsection, we choose parameters in line
with poliomyelitis  epidemiology. To isolate the effects of varying phase difference $\phi$ between the pulses, we set $\tau$ to 1 year and suppose that the patches are ``identical'': $b_1=b_2=b, \mu_1=\mu_2=\mu, \beta_1(t)=\beta_2(t)=\beta(t), \gamma_1=\gamma_2=\gamma, f_1=f_2=f, m_1=m_2=m,k_1=k_2=m$ and $\psi_1=\psi_2=\psi$.   Then, with no seasonality (i.e., $\beta(t)=\beta$ constant),
\begin{align*}
\widehat{R}_0\equiv \frac{\beta b}{\mu(\mu+\gamma)}  \tag{17}\label{r0hat}
\end{align*}
is the basic reproduction number of the autonomous (no pulse vaccination) version of model (\ref{2patch}) \cite{liu}.  We remark that, in this case of identical patches, the ``global'' reproduction number of the autonomous model, $\widehat{R}_0$, is equal to the ``patch'' reproduction number.  We also note that the above formula holds for the case $k_1=k_2\neq m$, but we assume that the migration rates for susceptibles and infected individuals are equal.  This assumption may be reasonable for polio, since at least 95\% of cases are asymptomatic \cite{WHOpolio}.

The reproduction number for polio
in an immunologically na\"ive population in low-income areas has been estimated in the range
6--14, although it has been speculated to be as high as 20 for some densely populated regions \cite{Fine}.  The mean infectious period for poliovirus is
approximately 16 days \cite{Fine}.  The average lifespan of individuals in the population is assumed to be 50 years.
Note that we will also vary the lifespan, in order to see its effect on $R_e$.
The total population size can be normalised to be 1 by letting
$b=\mu$.  Note that this can be done without loss of generality by dividing the equations in (\ref{2patch}) by $b/\mu$, which rescales the variables as fractions of carrying capacity (in the absence of migration) $b/\mu$, and rescales the parameters $b$ and $\beta$.  Explicitly, the following parameters will be used in the
following subsection: $b=\mu=1/50 $, $\gamma=365/16$, $\beta=319.655$.  These parameters give a value of $\widehat{R}_0=14$,
close to the upper bound of reproductive potential for poliovirus.
The pulse vaccination proportions, $\psi$, are taken to be $\psi=0.231$, in order to bring the reproduction number close to the threshold value of $1$.   This proportion may seem small, but it should be mentioned that vaccination campaigns only target children; therefore, $23.1\%$ represents a substantial percentage of children to vaccinate.
The coupling parameter values depend on the specific regions of consideration.  In the simulations below, the coupling parameters $m$ and $f$ will be varied and, in some instances, will be chosen relatively large to illustrate the effect of phase difference on $R_e$. Also, a seasonal transmission rate will be introduced in Section \ref{seasSec}.  The aforementioned parameter values are given in Table
\ref{table}, along with units and descriptions.

\begin{table}[t!]
\caption{Parameter values, units and description for case of identical patches} 
\centering 
\begin{tabularx}{\textwidth}{>{} lXXX}
\toprule
Parameter  &  Value & Units & Meaning\\ [0.5ex]
\toprule
\\
$b$ &  1/50 & ${\rm individuals}\times {\rm year}^{-1}$  & birth rate  \\ [0.5ex]	
$\mu$ &  1/50 & ${\rm year}^{-1}$  & death rate  \\ [0.5ex]
$\beta$ & 319.655 & $\left({\rm individuals} \times {\rm year}\right)^{-1}$ & transmission rate\\[0.5 ex] 
$\gamma$ & (1/(16/365)) & ${\rm year}^{-1}$ & recovery rate \\ [0.5 ex]
$m,k$ & varied; $0\leq m,k \leq 2$ &  ${\rm year}^{-1}$  & migration rate of susceptible and infected  \\ [0.5ex]
$f$ & varied; $0\leq f \leq 1/2$ & -- & fraction of cross transmission  \\[0.5 ex] 
$\psi$ & 0.231 & -- & pulse vaccination proportion \\[0.5 ex]
$\tau$ & 1 & years & period \\[0.5 ex]
$\phi$ & varied; $ 0\leq \phi \leq 1$ & ${\rm years}$ & phase difference between pulses \\ [0.5 ex]
$\widehat{R}_0$ & 14 & -- & reproduction number for na\"ive population \\ [0.5 ex] 
\bottomrule
\end{tabularx}
\label{table}
\end{table}

First, consider the case where linear migration is included without seasonality or mass-action coupling of the patches; i.e., $m>0$, $f=0$ and $\beta(t)=\beta=\text{constant}$.  Numerical calculations of $R_e$ as the phase difference between the pulses varies are presented in Figure \ref{fig:mig}.  Notice that $R_e$ is minimised when the pulses are in phase; i.e., when the patches synchronise their vaccination campaigns.  Also, as the migration rate $m$ increases, $R_e$ becomes increasingly sensitive to the phase difference, $\phi$.  These observations are consistent with Theorem \ref{R0theorem}. Also, there is an intuitive explanation for this result.  When regions coupled by migration employ pulse vaccination, we can think of how to best time the vaccination pulses in order to immunise as many migrants as possible.  If the pulses are de-synchronised, then it is possible --- even with 100\% coverage in each patch --- that a migrant (who is born at some time $t_0$ and does not die during the period $[t_0,t_0+\tau]$) can remain unvaccinated by being in Patch 2 when Patch 1 employs pulse vaccination and in Patch 1 when Patch 2 conducts their pulse vaccination.  There is evidence that this effect has led to measles epidemics in the coupled regions of Burkina Faso and C\^ote d'Ivoire in Africa \cite{measles}.   Synchronising the pulses can most effectively reach the migrant population.  Indeed, when the average total susceptible population over the year is plotted with respect to $\phi$, the graph has the same shape as Figure \ref{fig:mig}.   In other words, synchronising the pulses will produce the highest time-averaged coverage for a fixed proportion, $\psi$, of susceptibles that can be vaccinated in each pulse.  This is of course expected from Theorem \ref{R0theorem}, where we prove these statements about synchronisation (locally and in a more restricted setting).  

\begin{figure}[t]
\subfigure[][]{\label{fig:miga}\includegraphics[width=9cm]{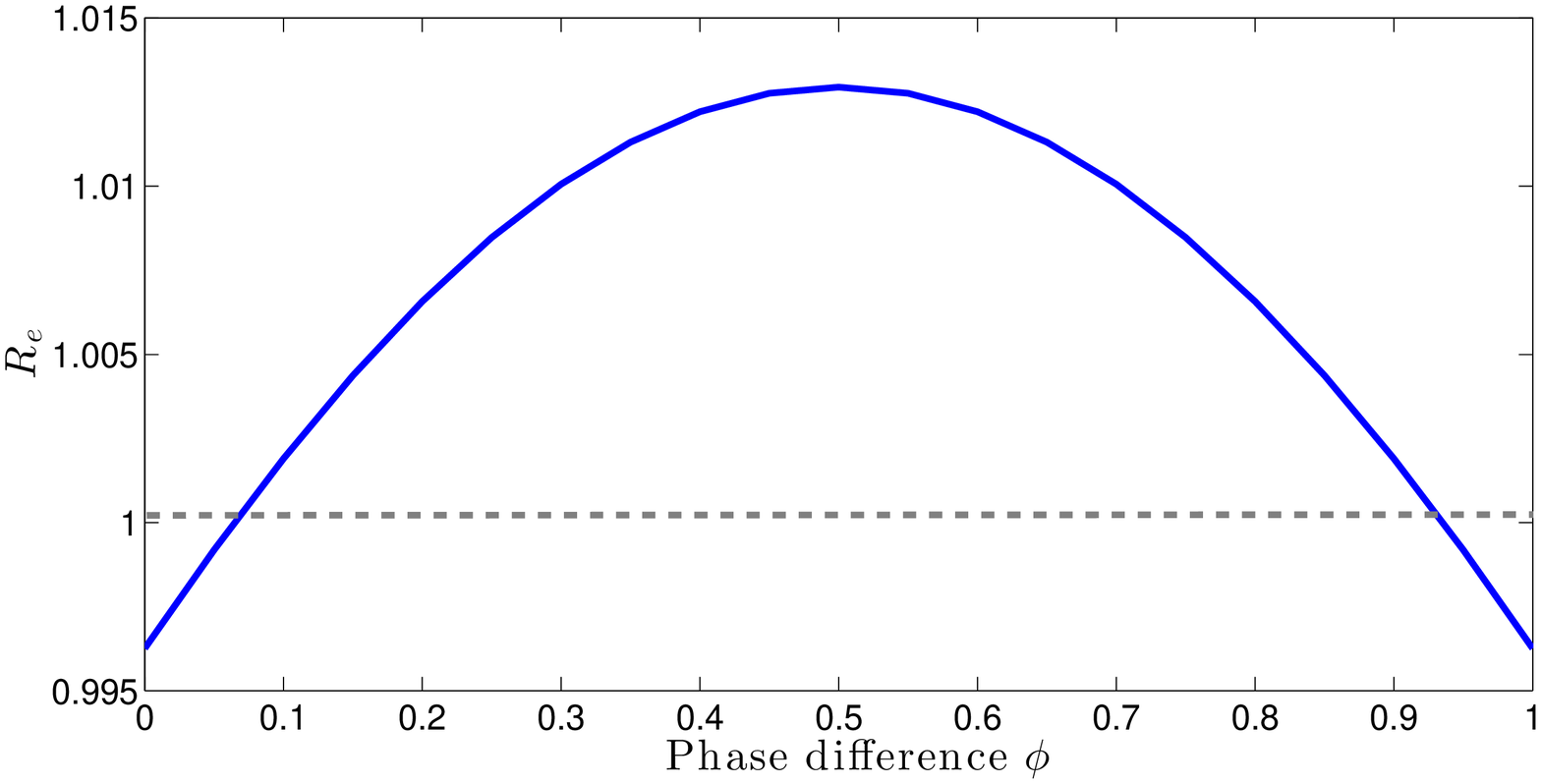}}
\subfigure[][]{\label{fig:migb}\includegraphics[width=9cm]{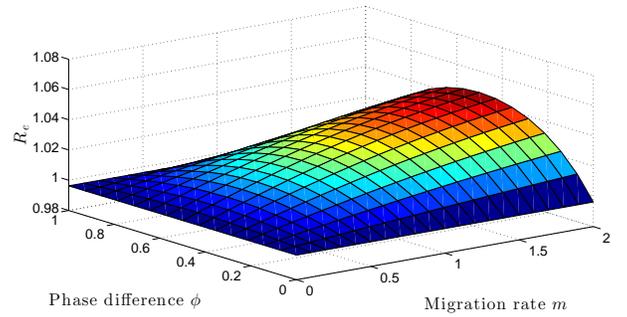}}
  \caption{The effects of phase difference when the transmission rate is constant and there is only migration.  (a) $R_e$ vs phase difference $\phi$, for the case $m=0.5$.  (b) $R_e$ vs phase difference $\phi$ and migration rate $m$.}
    \label{fig:mig}
\end{figure}

Next, suppose that the patches are only coupled through mass-action cross transmission without seasonality; i.e., $0\leq f\leq 1/2$, $m=0$ and $\beta(t)=\beta$.  Figure \ref{fig:cross} displays numerical calculations of $R_e$ versus the phase difference $\phi$ for this case.  Again, $R_e$ is always minimised when the pulses are synchronised; i.e., $\phi=0$.  However, this case is more subtle than the previous one.  When the average total susceptible population over the year is taken as a function of the phase difference $\phi$, it is not hard to see that this will be constant as $\phi$ varies between $0$ and $1$.  Thus the optimality of pulse synchronisation cannot be explained like the previous case where the (averaged) susceptible population was minimised when $\phi=0$, and Theorem \ref{R0theorem} cannot be applied.  Also, observe that the phase difference becomes a non-factor as $f\rightarrow 1/2$ in Figure \ref{fig:crossb}.  In this case, the contribution of cross transmission becomes equal to within-patch transmission when $f\rightarrow 1/2$, causing the infected in a patch to have equal magnitude of correlation with either pulse.  This is likely the reason that the phase difference does not affect $R_e$ when $f=1/2$.

\begin{figure}[t]
\subfigure[][]{\label{fig:crossa}\includegraphics[width=9cm]{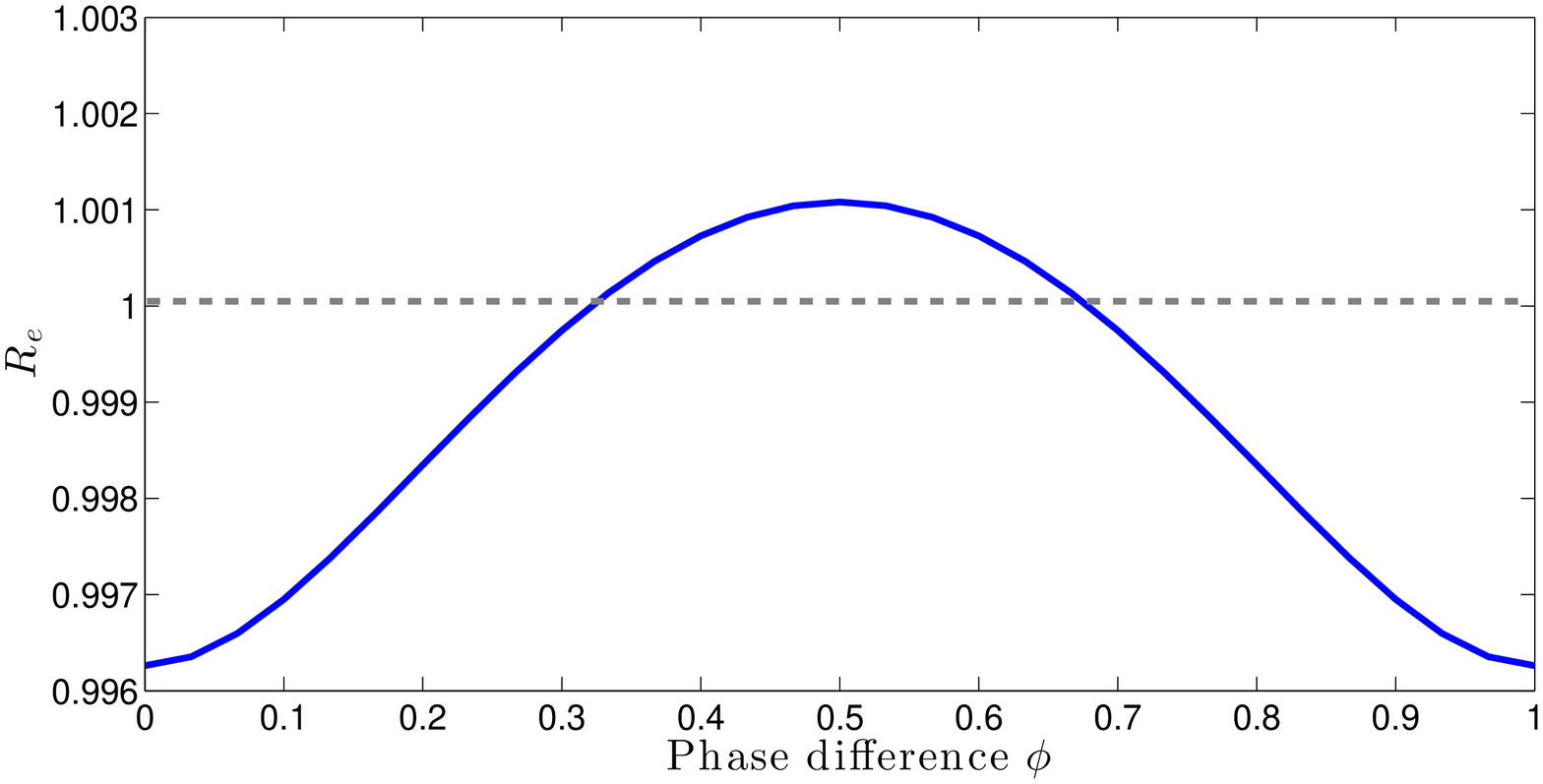}}
\subfigure[][]{\label{fig:crossb}\includegraphics[width=9cm]{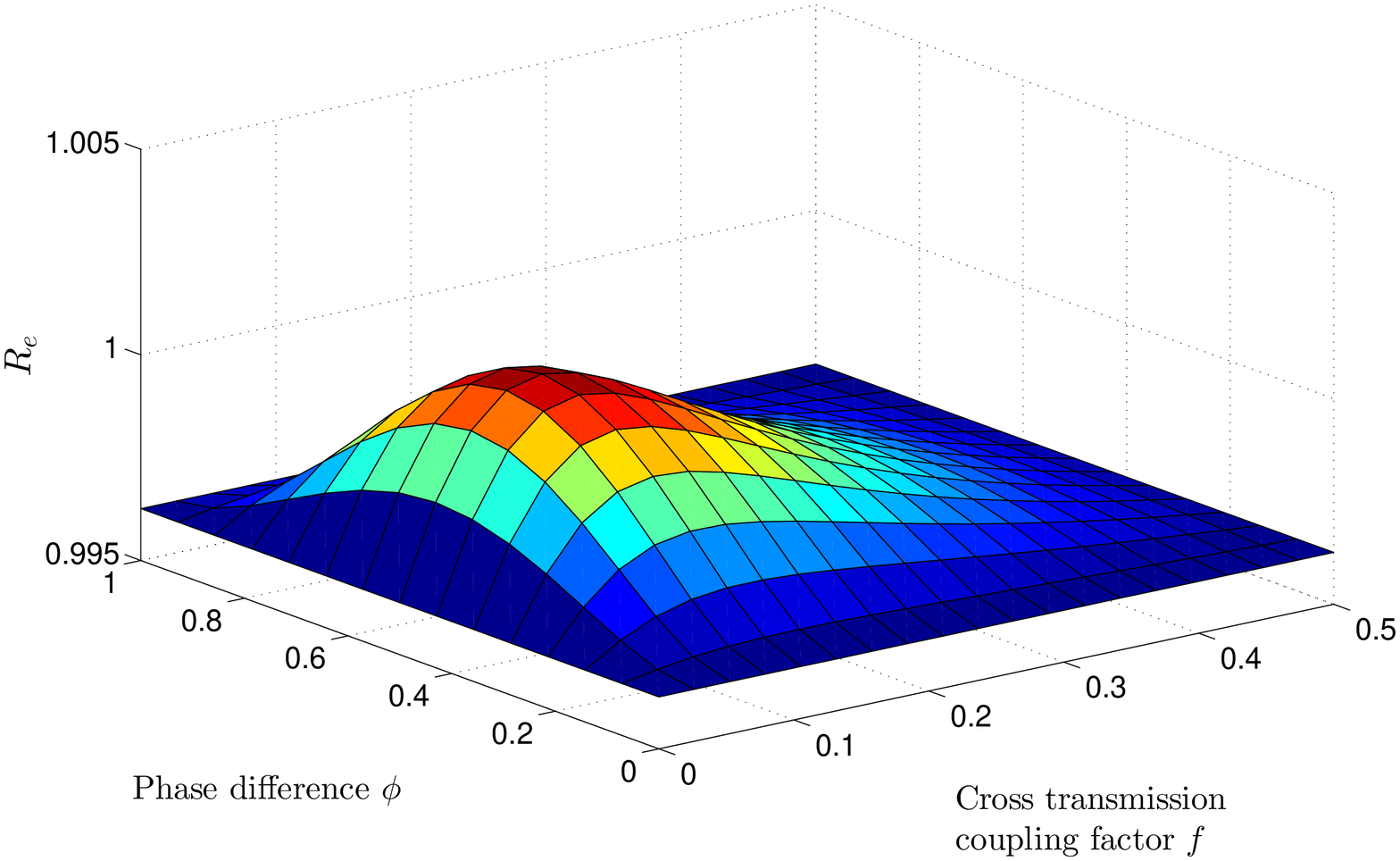}}
  \caption{The case when identical patches are coupled through mass-action cross transmission without seasonality.  (a) $R_e$ vs phase difference $\phi$ for the coupling factor $f=0.1$.  (b) $R_e$ vs phase difference $\phi$ and coupling factor $f$.}
  \label{fig:cross}
  \end{figure}
  
Theorem \ref{R0theorem} implies that $R_e(\phi)$, the reproduction number $R_e$ as a function of phase difference $\phi$, may be most sensitive to $\phi$ when $m$, $\widehat{R}_0$, $\psi$ and $\mu$ are large; this is confirmed in simulations.  From Figures \ref{fig:migb} and \ref{fig:crossb}, it is seen that the migration rate $m$ and coupling factor $f$ strongly affect the amplitude of $R_e(\phi)$.  If the migration rate $m$ is large or if $f$ is close to a certain value (around 0.1 in Figure \ref{fig:crossb}), pulse synchronisation becomes increasingly important, since $R_e$ can vary largely with $\phi$.  In Figure \ref{fig:r0parama}, observe that, as $\mu$ and $\psi$ increases, while keeping $b=\mu$ and fixing the other parameters, the amplitude of $R_e(\phi)$ increases. In Figure \ref{fig:r0paramb}, we plot the pulse vaccination proportion $\psi$ required for $R_e=1$ as a function of $\phi$ for three different wild (before immunisation) reproduction numbers, $\widehat{R}_0$.  As $\widehat{R}_0$ increases, more vaccination is required to bring $R_e$ to unity and the ``phase effect'' increases. For the cases where the amplitude of $R_e(\phi)$ is relatively large, it is vital to synchronise the pulses since the parameter $\phi$ can be the difference between extinction and persistence of the pathogen.  In Figure \ref{fig:deter}, there are simulations of the system (\ref{2patch}) in the case of in-phase pulses ($\phi=0$), resulting in eradication, and out-of-phase pulses ($\phi=0.5$), resulting in disease persistence.

 \begin{figure}[t!]
\subfigure[][]{\label{fig:r0parama}\includegraphics[width=9cm]{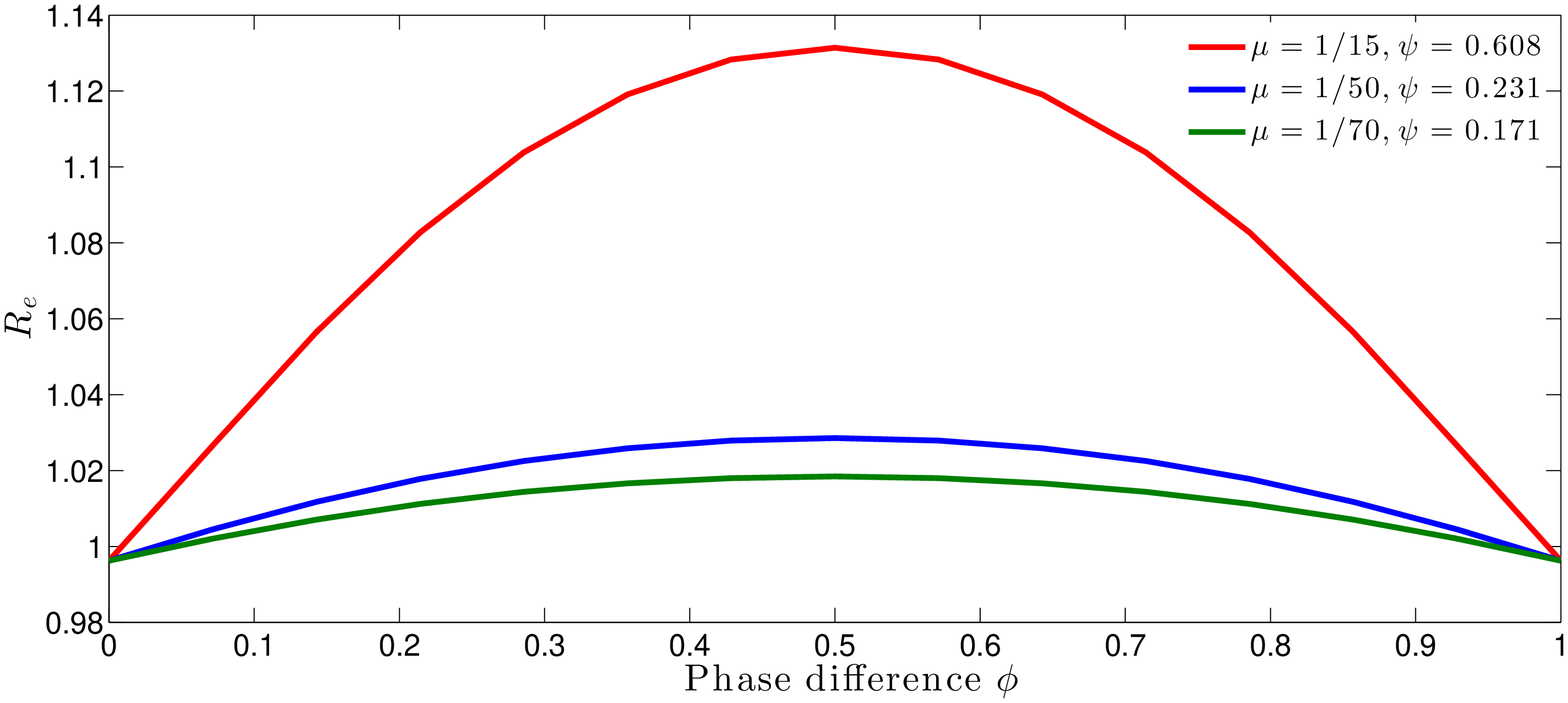}} 
\subfigure[][]{\label{fig:r0paramb}\includegraphics[width=9cm]{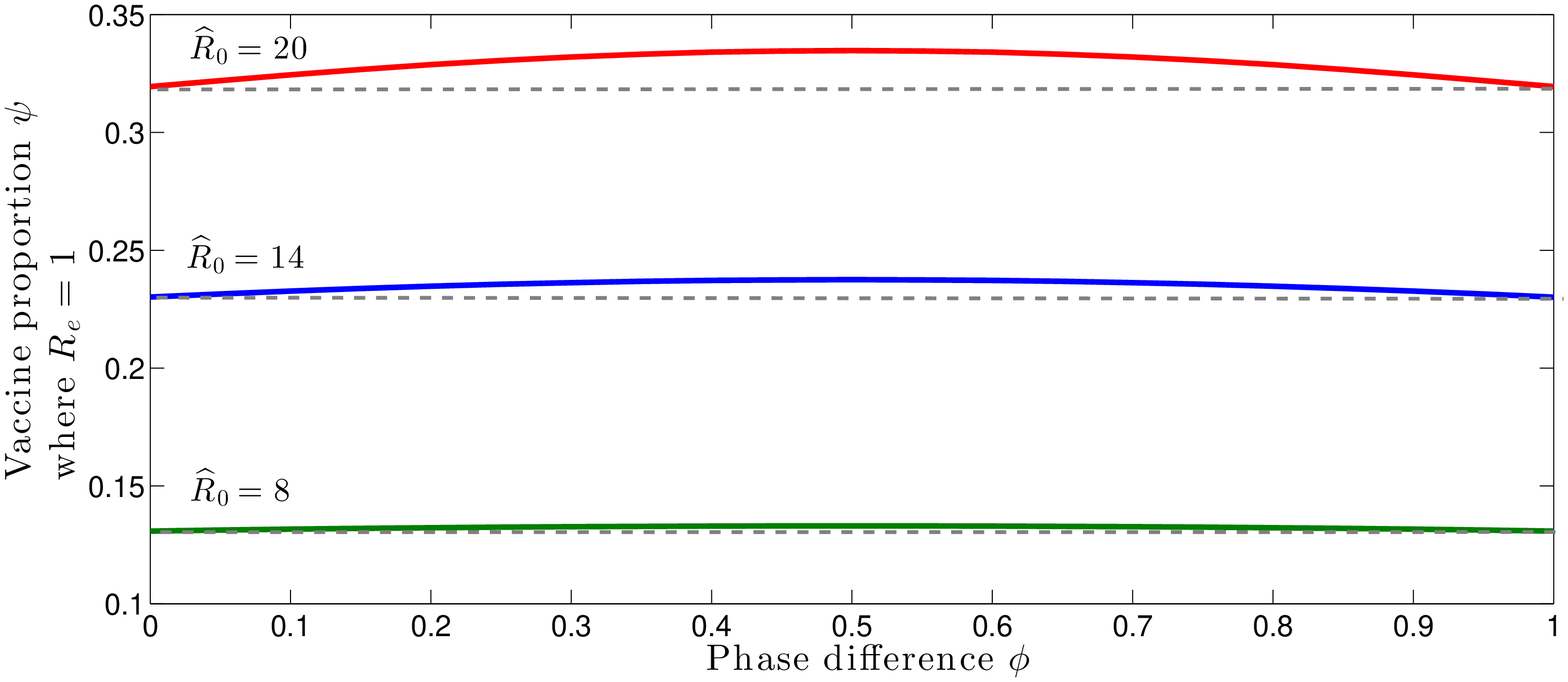}}
\caption{ The effect of $\mu$, $\psi$ and $\widehat{R}_0$ on $R_e(\phi)$.  (a)  $R_e(\phi)$ for three different values of $\mu$ and $\psi$ with $m=1$, $f=1$ and all other parameters as in Table \ref{table}.  (b) The pulse vaccination proportion $\psi$ required for $R_e=1$ as the phase difference $\phi$ varies for three different values of $\widehat{R}_0$ with $m=1$, $f=1$ and all other parameters as in Table \ref{table}.}
  \label{fig:r0param}
  \end{figure}

 \begin{figure}[t!]
\subfigure[][]{\label{fig:detera}\includegraphics[width=9cm]{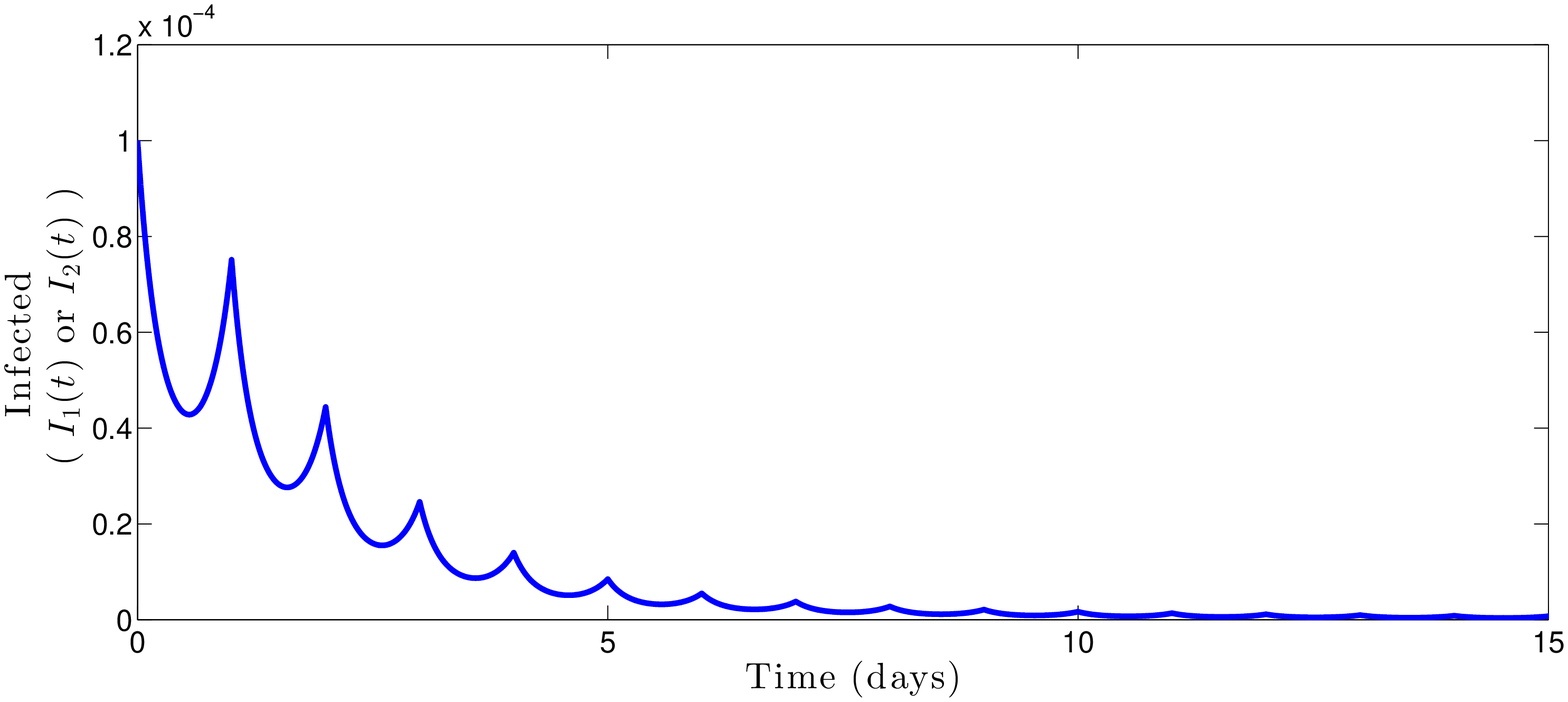}}
\subfigure[][]{\label{fig:deterb}\includegraphics[width=9cm]{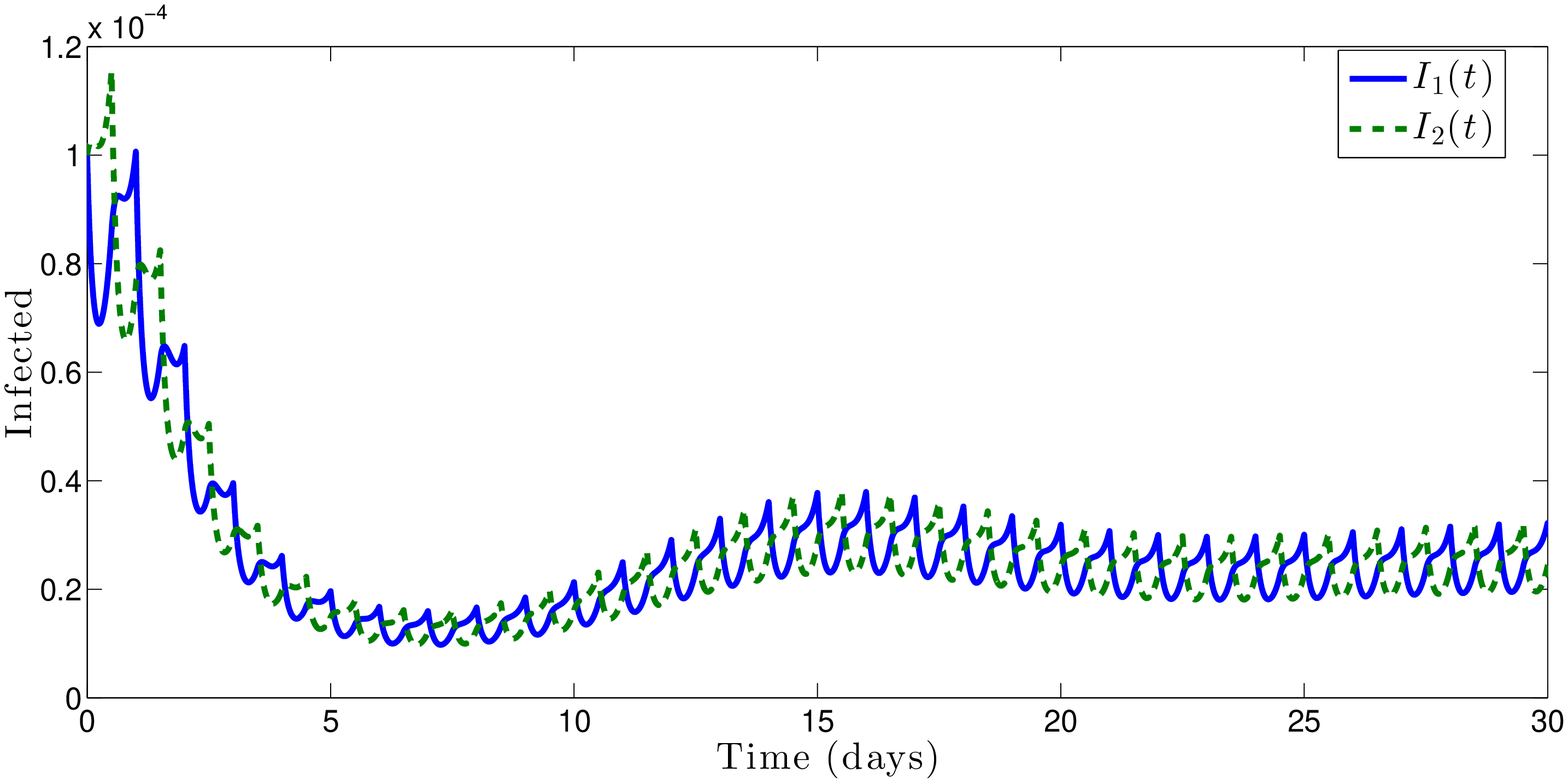}}
\caption{Simulations showing how the timing between the pulse vaccination can determine whether the disease persists. The coupling parameters are taken to be $m=1$ and $f=0.1$. (a) Simulation of the infected in Patches 1 and 2 when the vaccination pulses are synchronised; i.e., $I_i(t)$ when $\phi=0$. (b) Simulation of the infected in Patches 1 and 2 when the pulses are completely desynchronised; i.e., $I_i(t)$ when $\phi=0.5$. Identical patches are considered.}
  \label{fig:deter}
\end{figure}

\subsection{The SIR model with identical patches and seasonality}\label{seasSec}
Now consider identical patches with seasonality, where $\beta(t)=\beta(1+a\sin(2\pi (t-\theta)))$.  Here $\theta$ is a seasonal phase-shift parameter, which allows us to vary the timing of the pulses throughout the year.  In the example simulated in Figure \ref{fig:seas}, it is optimal to synchronise the pulse vaccinations and to execute just before the high-transmission season.  This finding  agrees with results obtained for single patch SIR models \cite{optimalvac}.  The importance of synchronising the pulses increases with migration rate $m$, while the sensitivity to timing the pulses with respect to seasonality increases with the seasonal forcing amplitude $a$.  Also, the sensitivity to timing pulses with respect to each other and seasonality both increase with $\psi$, $\mu$, and $\widehat{R}_0$. 

\begin{figure}[t!]
\centering
\includegraphics[width=9cm]{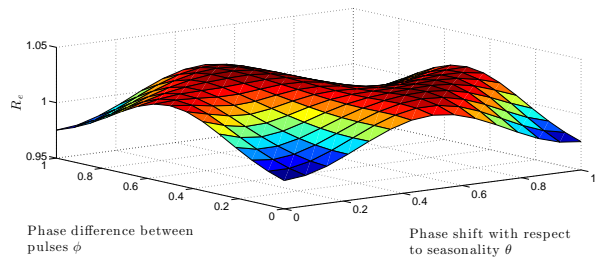}
 \caption{$R_e$ vs phase difference between pulses for the system with seasonality.   The seasonal transmission is of the form $\beta(t)=\beta(1+a\sin(2\pi (t-\theta)))$ where $\theta$ is the seasonal phase shift.  Here $a=0.5$, the migration rate is $m=0.5$, there is no cross-transmission ($f=0$), and the other parameters are as in Table \ref{table}.  In this case, the results show that it is best to synchronise pulse vaccinations and to execute them during the season before the high-transmission season. }
 \label{fig:seas}
 \end{figure}

If the seasonal transmission coefficients for the two patches are not in phase, then optimal timing of pulses with respect to seasonality can be in conflict with synchronising the pulses.  This creates a trade-off between synchronising the pulses and optimally timing the pulse in each patch according to the transmission season.  In the pulse vaccination operation against polio, Operation MECACAR, public health officials had to consider this trade-off \cite{NID}.  In this case, they decided that pulse synchronisation was most important.  Theoretically, the optimal timing of pulse vaccinations should depend on the specific parameters, especially the relative size of migration rate to seasonal forcing amplitude.  To illustrate this phenomenon, we consider transmission rates $\beta_1(t) =\beta(1+a\sin(2\pi (t-\theta)))$ and $\beta_2(t)=\beta(1+a\sin(2\pi (t-\theta-\sigma)))$ for Patch 1 and Patch 2, respectively.  Here $\sigma$ is the phase difference between the seasonal transmission rates of Patch 1 and Patch 2.  In Figure \ref{fig:seasop}, $R_e$ is calculated for the case where the seasonal transmission rates are out of phase; i.e., $\sigma=0.5$.    In Figure \ref{fig:seasopa}, the migration rate $m$ is set to $0.5$ and the mass-action coupling $f$ is $0$.  For this case, the seasonal transmission has a larger effect than the migration, and it is best to desynchronise the pulses so that each pulse occurs in the season before the higher transmission season.  In Figure \ref{fig:seasopb}, the migration rate is assumed to be larger ($m=2$); in this scenario, it is best to synchronise the pulses.

\begin{figure}[t!]
\subfigure[][]{\label{fig:seasopa}\includegraphics[width=9cm]{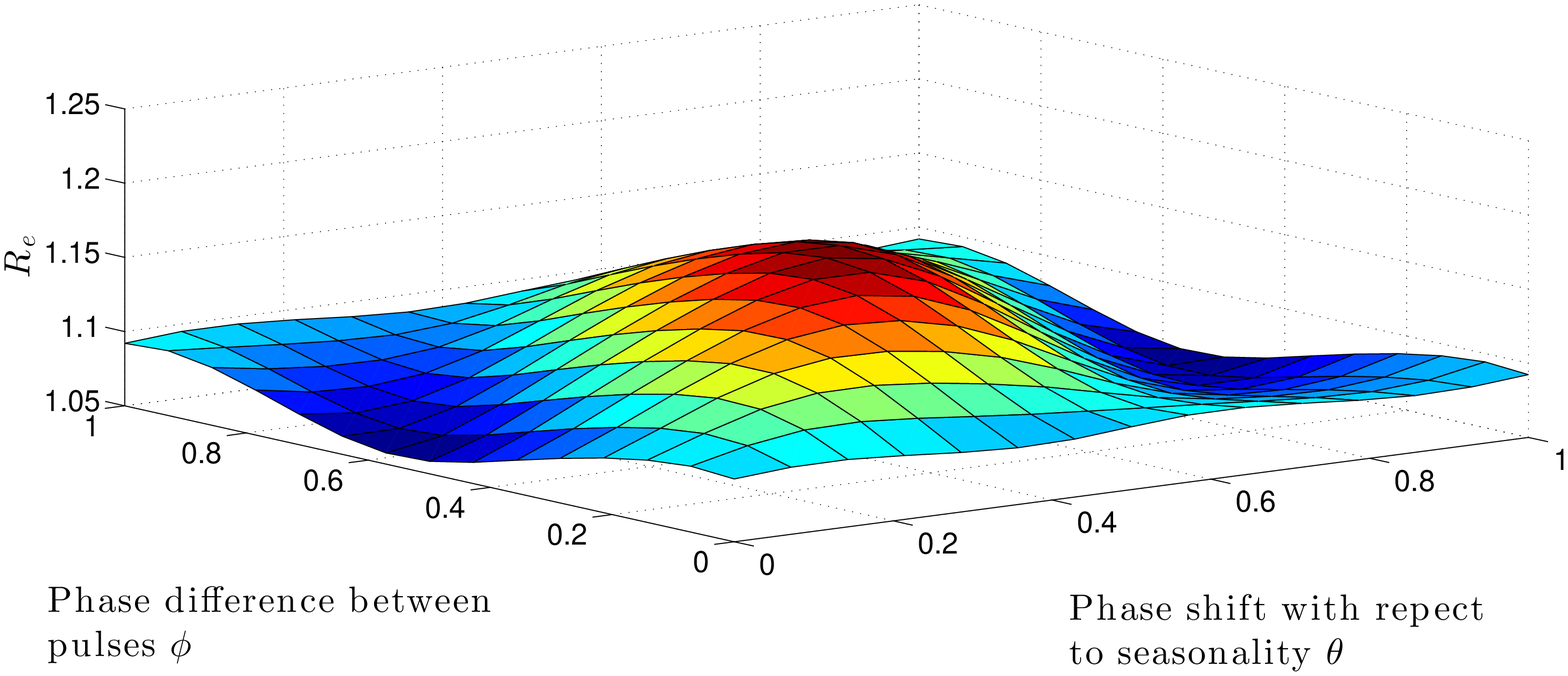}}
\subfigure[][]{\label{fig:seasopb}\includegraphics[width=9cm]{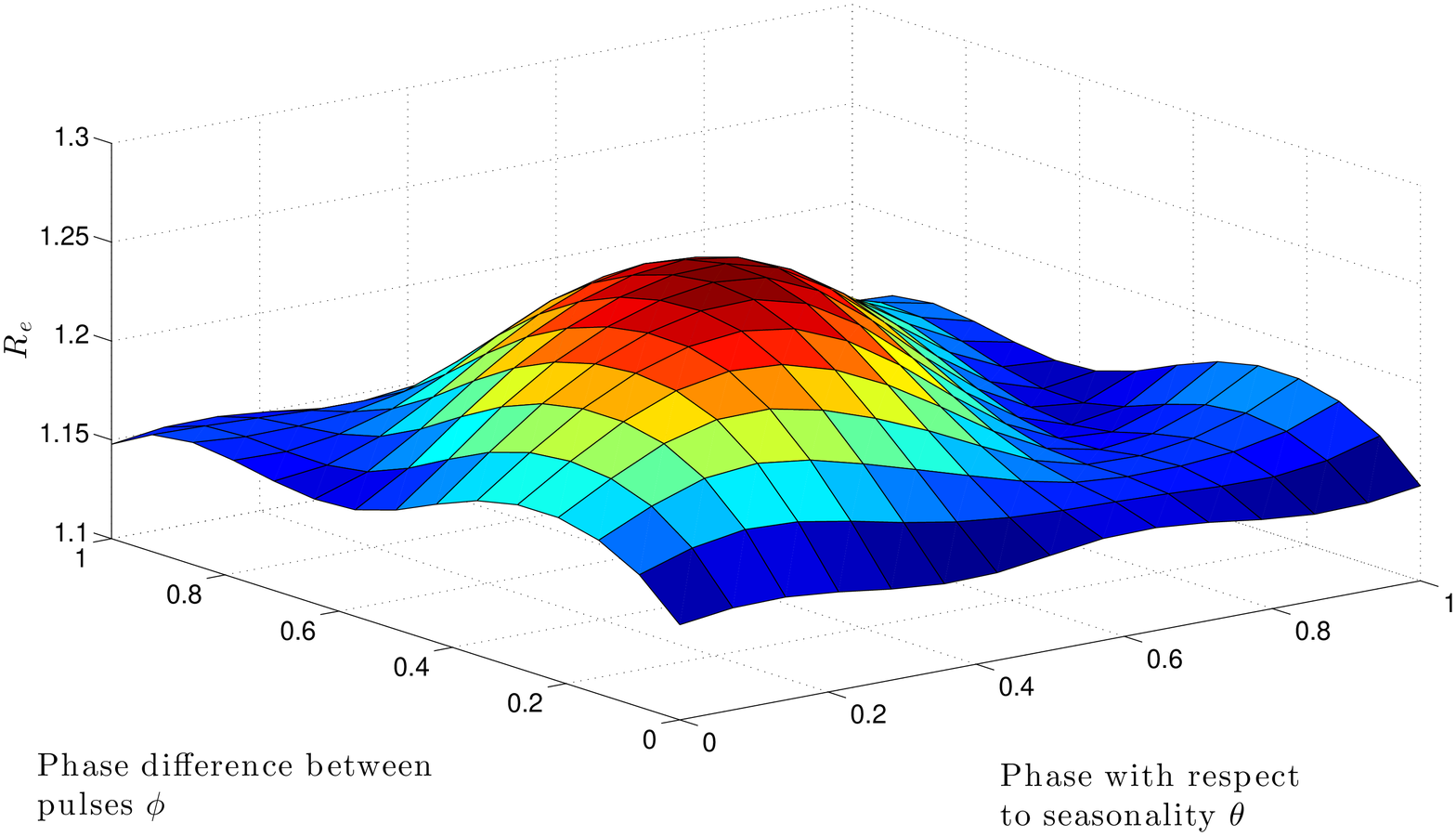}}
 \caption{$R_e$ vs phase difference between pulses for the system with out-of-phase seasonal transmission rates.   The seasonal transmission rates are of the form $\beta_1(t)=\beta(1+0.5\sin(2\pi (t-\theta)))$ and $\beta_2(t)=\beta(1+0.5\sin(2\pi (t-\theta-0.5)))$.   In (a), the migration rate, $m$, is set to $0.5$.  For this case, it is best to desynchronise pulses.  In (b), the migration rate is set to $2$, and it is best to synchronise the pulses.  In both figures, there is no mass-action coupling $(f=0)$ and the other parameters are specified in the text. }
 \label{fig:seasop}
 \end{figure}

 \subsection{Comparison of vaccination strategies and effect of different movement scenarios on optimal vaccine distribution}
An interesting and possibly applicable exercise is to compare a constant-vaccination strategy with the pulse-vaccination strategy.  From a theoretical standpoint, it is important to reconcile results obtained for pulse vaccination with the findings for a smooth, constant vaccination rate.  On the practical side, disease-control authorities may like to know the optimal vaccination strategy based on a simple cost measure.  The basic measure that will be used to quantify the cost of a vaccination strategy is vaccinations per period $\tau$ calculated at the disease-free periodic solution.  From an economic perspective, this cost measure has the appeal of simplicity.  To understand why this definition can also be the dynamically sound way of measuring cost, it is instructive to consider the case of isolated patches or, without loss of generality, a single patch under a general periodic vaccination strategy.  Specifically, consider the following system:
\begin{align*}
\frac{dS}{dt}&=b-\mu S-\beta(t) SI - \zeta(t)S \\
\frac{dI}{dt}&=\beta(t) SI-(\mu+\gamma)I,
\end{align*}
where the $\zeta(t)$ is a $\tau$-periodic vaccination rate and the transmission rate, $\beta(t)$, is $\tau$-periodic.  In the case of constant vaccination, $\zeta(t)\equiv \sigma$, where $\sigma\in \mathbb{R}_+$.  For pulse vaccination, $\zeta(t)=\sum_{n\in\mathbb{N}}\psi \delta_{n\tau}$, where $\delta_t$ is the Dirac delta mass centred at $t$ and $0\leq \psi\leq 1$.  In \cite{optimalvac}, the authors rigorously define the appropriate space of periodic vaccination rates to include the Dirac delta mass and guarantee existence of a unique disease-free susceptible periodic solution, $\overline{S}_{\zeta}(t)$, for any periodic vaccination rate $\zeta(t)$ in this setting.  The cost of vaccination (vaccinations per period calculated at $\overline{S}_{\zeta}(t)$) is 
\begin{align*}
C_{\zeta}\equiv  \int_0^{\tau} \zeta(t)\overline{S}_{\zeta}(t)\,dt.
\end{align*}
Using the next-generation characterisation (\ref{rcomp}), the effective reproduction number, $R_e$, can be explicitly found as
\begin{align*}
R_e=\frac{1}{\mu+\gamma}\frac{1}{\tau} \int_0^{\tau} \beta(t)\overline{S}_{\zeta}(t) \,dt. 
\end{align*}

Onyango and M\"uller studied optimal vaccination strategies in this model in terms of minimising $R_e$ \cite{optimalvac}.  Here we give a simple representation of $R_e$ that can yield insight into comparing vaccination strategies, but do not provide the rigorous construction of the optimal strategy done by Onyango and M\"uller \cite{optimalvac}.  Specifically, we rewrite $R_e$ for a general periodic vaccination strategy $\zeta(t)$ in a form that compares it to the constant-vaccination strategy of equal cost.  First, as noted in \cite{optimalvac}, by integrating the $\dot{S}$ equation over one period, the following can be obtained:
\begin{align*}
C_{\zeta}=\tau b-\mu \int_0^{\tau} \overline{S}_{\zeta}(t)\,dt.
\end{align*}

Define the average transmission rate as $\langle\beta\rangle=\frac{1}{\tau}\int_0^{\tau}\beta(t)\,dt$.  For constant vaccination, $\zeta_1(t)\equiv\sigma$, so we find that $C_{\zeta_1}=\frac{\sigma b \tau}{\mu+\sigma}$ and the effective reproduction number is $R_e^c=\frac{\langle\beta\rangle b}{(\mu+\sigma)(\mu+\gamma)}$.  For the periodic vaccination rate $\zeta(t)$, we rewrite the effective reproduction number $R_e$ by comparing it to a constant-vaccination strategy of equal cost:
\begin{align*}
\langle\beta\rangle\frac{\sigma b \tau}{\mu+\sigma}&=\langle\beta\rangle C_{\zeta_1}=\langle\beta\rangle C_{\zeta}=\langle\beta\rangle   \tau b - \langle\beta\rangle\mu \int_0^{\tau}\overline{S}_{\zeta}(t)\,dt \\
\langle\beta\rangle\frac{\sigma b \tau}{\mu+\sigma}&= \langle\beta\rangle\tau b -\mu \int_0^{\tau}\beta(t)\overline{S}_{\zeta}(t)\,dt+\mu \int_0^{\tau}\overline{S}(t)(\beta(t)-\langle\beta\rangle)\,dt \\
\langle\beta\rangle\frac{\sigma b \tau}{\mu+\sigma}&= \langle\beta\rangle\tau b -\mu(\mu+\gamma)\tau R_e+\mu \int_0^{\tau}\overline{S}_{\zeta}(t)(\beta(t)-\langle\beta\rangle)\,dt \\
 \Leftrightarrow R_e&=R_e^c +\frac{1}{(\mu+\gamma)\tau}\int_0^{\tau} \overline{S}_{\zeta}(t)(\beta(t)-\langle\beta\rangle)\,dt .
\end{align*}

If we normalise $\overline{S}_{\zeta}(t)$ by letting $\bar s_{\zeta}(t)=\frac{\mu}{b}\overline{S}_{\zeta}(t)$ and denote $\widehat{R}_0=\frac{\langle\beta\rangle b}{\mu(\mu+\gamma)}$ (the reproduction number in the absence of vaccination), then the following is obtained:
\begin{align*}
R_e=R_e^c-\widehat{R}_0 \frac{1}{\tau} \int_0^{\tau} \overline{s}_{\zeta}(t)\left(1-\frac{\beta(t)}{\langle\beta\rangle}  \right) \,dt.
\end{align*}

Clearly, if $\beta(t)$ is constant --- i.e., $\beta(t)=\langle\beta\rangle$ --- then all vaccination strategies are equivalent, in particular pulse- and constant-vaccination strategies, and $R_e=R_c$.  This observation provides justification as to why $C_{\zeta}$ is an appropriate cost measure from an epidemiological point of view.  When $\beta(t)$ is not constant, then a different result is obtained.  Define $\alpha(t)=1-\frac{\beta(t)}{\langle\beta\rangle}$ and notice that $\overline{\alpha}=0$.  Then $\overline{s}_{\zeta}(t)$ acts as a weighting function and can be chosen to maximise $\int_0^{\tau}\alpha(t) \overline{s}_{\zeta}(t)$, thereby minimising $R_e$.  Intuitively, a susceptible profile $\overline{s}_{\zeta}(t)$ that is minimal for the range of values where $\alpha(t)<0$ and maximal for $\alpha(t)>0$ would seem to work the best.  The rigorous construction of the optimal vaccination strategy was carried out by Onyango and M\"uller \cite{optimalvac}.  They found that a single, well-timed pulse is the optimal strategy (assuming that the allotted cost can be exhausted by a single pulse; otherwise, the optimal susceptible profile requires $\overline{s}(t)\equiv 0$ on $[t_1,t_2]\subset [0,
\tau)$ and no vaccinations can occur in $[0,\tau)\setminus [t_1,t_2]$) \cite{optimalvac}.  For the pulse-vaccination strategy $\zeta(t)=\sum_{n\in\mathbb{N}}\psi \delta_{n\tau}$, we find that 
$$\overline{s}_{\zeta}(t)=1-e^{-\mu t}\left(1-\frac{(1-\psi)(1-e^{-\mu \tau})}{1-e^{-\mu \tau}(1-\psi)}\right). $$
It can be inferred that the advantage of the optimal pulse vaccination over constant vaccination (in terms of difference in $R_e$) increases with $\widehat{R}_0$, $\mu$ (when $\widehat{R}_0$ remains fixed), $\psi$ and the amplitude of $\alpha(t)$.  

A natural question to ask is whether similar results can be obtained for the two-patch model.  First, for the case of no seasonality, does the equivalence of vaccination strategies hold?  The cost can still be defined as number of vaccinations per period in each patch calculated at the disease-free periodic solution.  The vaccination rate $\zeta(t)=(\zeta_1(t),\zeta_2(t))^T$ has two components.  Consider the diagonal matrix $Z(t)={\rm diag}(\zeta_1(t),\zeta_2(t))$ and the $2\times 1$ disease-free periodic solution vector $\overline{S}_{\zeta}(t)$.  Then $C_{\zeta}= \int_0^{\tau}Z(t)\overline{S}_{\zeta}(t)\,dt$; here $C_{\zeta}$ is a $2\times 1$ vector containing the cost in each patch.  The effective reproduction number, $R_e$, is defined in Section \ref{sec4}, but cannot be explicitly expressed for multiple patches.  Using notation from Section \ref{sec4}, we note that the reproduction number, $R_e$, for constant vaccination in the case of no seasonality (i.e., $F(t)=F,V(t)=V$ are independent of time) is found to be $FV^{-1}$, which agrees with the next-generation matrices for autonomous disease-compartmental models \cite{wang}.  

Under the conditions of Theorem \ref{R0theorem} --- no cross-transmission and no migration of infected --- the equivalence of vaccination strategies of equal cost holds for constant transmission rate by (\ref{R0diag}). Numerical simulations showed that shifting the phase difference between the pulse vaccinations can alter the value of $R_e$ when there is cross transmission and no migration (Figure \ref{fig:cross}), even though the cost $C_{\zeta}$ remains constant.  Thus the equivalence of vaccination strategies of equal cost cannot hold for the general constant transmission case in system (\ref{2patch}).  However, we performed simulations with many different parameters showing that the synchronised pulse vaccinations have values of $R_e$ very close (within the range of numerical error) or identical to the reproduction number for the constant-vaccination strategy of equal cost.  For the case of no cross transmission (but migration of \emph{both} susceptible and infected), simulations produced identical or nearly identical reproduction numbers for constant and pulse vaccinations of equal cost, independent of the phase difference $\phi$.  This is not in contradiction with Figure \ref{fig:mig}, since shifting the phase difference $\phi$ in the migration model alters the cost of vaccination; i.e., synchronised pulses result in more vaccinations per period than desynchronised pulses in model (\ref{2patch}) when $m_1,m_2>0$.

 \begin{figure}[t!]
\subfigure[][]{\label{fig:optdeploya}\includegraphics[width=9cm]{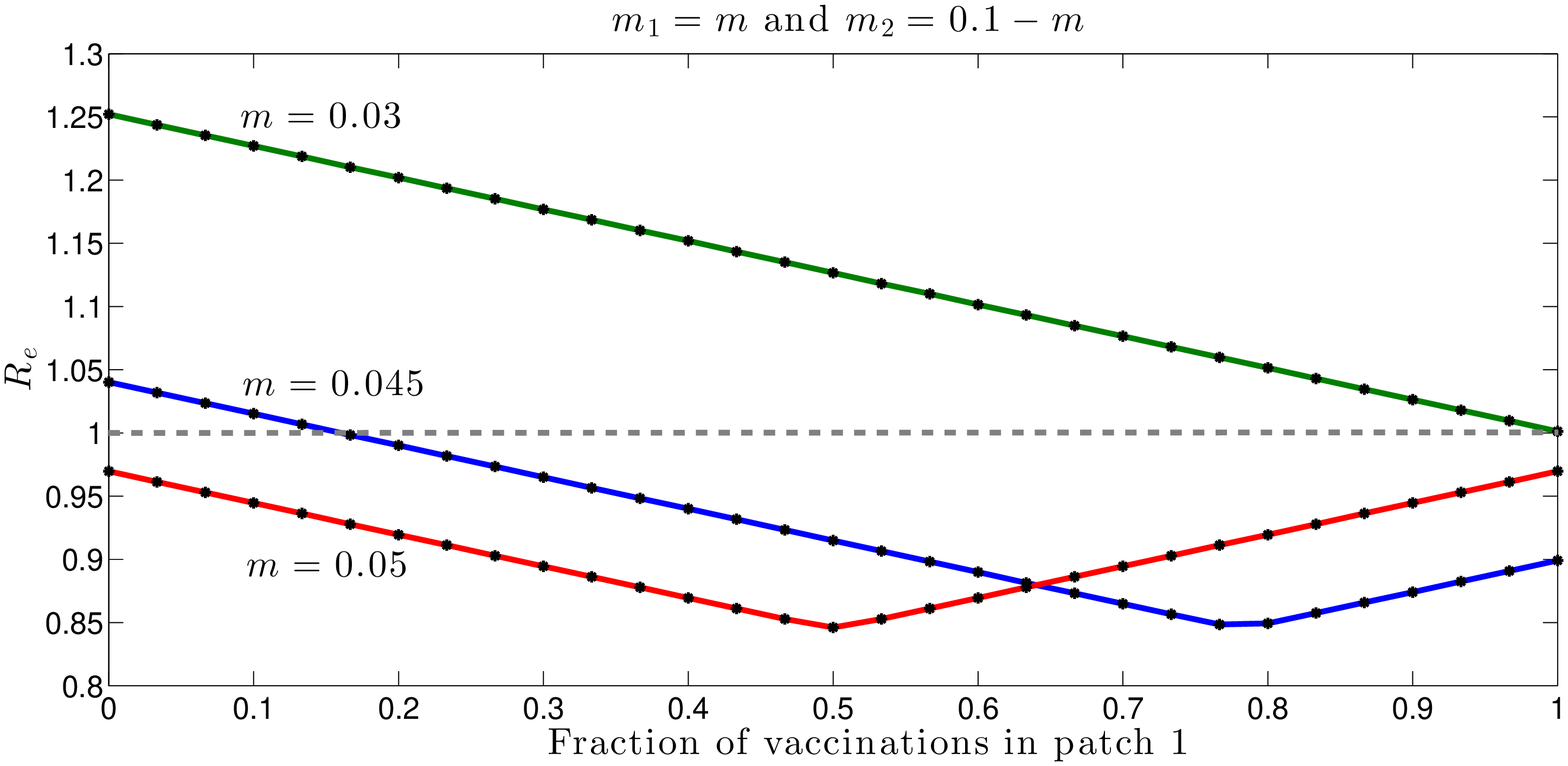}} 
\subfigure[][]{\label{fig:optdeployb}\includegraphics[width=9cm]{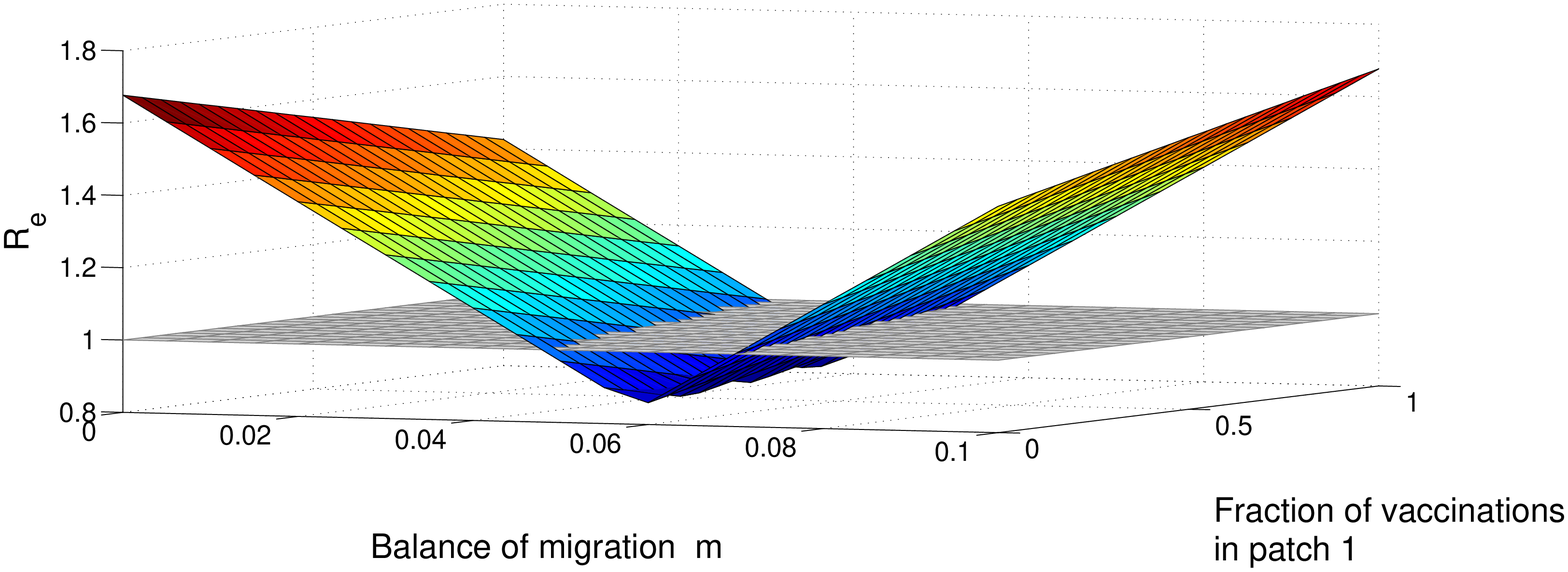}}
\caption{Simulations showing that an imbalance in migration rates can affect $R_e$ and the optimal deployment of the vaccine in otherwise identical patches.  The parameters are $\beta=36.5$, $\phi=0$, $f=0$, and $b,\mu,\gamma$ as in Table \ref{table}, yielding $\widehat{R}_0=1.6$.  The migration rates are $m_1=m$ and $m_2=0.1-m$, where $0\leq m\leq 0.1$.  The total amount of vaccinations per year, $V_{tot}$ (summed over both patches) --- i.e., $V_{tot}=\vec{1}\cdot C_{\zeta}$ --- is fixed at $V_{tot}=0.94\%$ of the total population, while the fraction, $q$, of the total allotted vaccinations distributed in Patch 1 varies.  In other words, $C_{\zeta}=\left(  qV_{tot},  (1-q)V_{tot} \right)^T$, where $0\leq q \leq 1$.  Note that the pulse vaccination proportions $\psi_1$ and $\psi_2$ vary as $q$ and $m$ change.  (a) $R_e$ vs $q$, for three different balance of migration rates: $m=0.03$, $m=0.045$ and $m=0.05$, along with $R_e=1$ (dashed line).  The solid coloured lines represent pulse vaccination, whereas the yellow stars represent constant vaccination of equal cost (notice that they are identical).  Notice also that the ratio $\frac{m_1}{m_2}$ affects $R_e$ for each vaccination scenario and alters the optimal vaccine distribution.  (b) A three-dimensional graph of $R_e$ vs $q$ and $m$ with the plane $R_e=1$.  This shows that the required distribution of vaccine needed to control the disease changes as the balance of migration changes.}
  \label{fig:optdeploy}
  \end{figure}

One implication of the cases where pulse vaccination and constant vaccination of equal cost agree on the value of $R_e$ is that results obtained in prior work on the autonomous multi-patch model can carry over to the impulsive model.  For example, an imbalance in migration rates has been shown to strongly affect $R_e$ in previous work on metapopulation models \cite{xiao}.  The same result is found in the case of pulse vaccination, as shown in Figure \ref{fig:optdeploy}.  For otherwise identical patches, an imbalance in migration rates $(m_1\neq m_2)$ causes the susceptibles and infected to concentrate more heavily in one patch, which increases the overall effective reproduction number.  This affects how the vaccine should be optimally distributed among the two patches, as illustrated in Figure \ref{fig:optdeploy}.

As in the single-patch model, including seasonality induces an advantage of well-timed pulse vaccination over constant vaccination of equal cost.  In Figure \ref{fig:compare}, we see that, as the amplitude of seasonality increases, synchronous pulse vaccinations applied the season before the high-transmission season can become more and more advantageous.  Simulations also show that the migration rate does not affect $R_e$ for the case of identical patches and simultaneous pulses.  The advantage of pulse vaccination over constant vaccination depends on the parameters, as detailed previously.  For the simulations in Figure \ref{fig:compare}, pulse vaccination can offer a substantive advantage over constant vaccination.  Hence the inherent advantage of pulse vaccination in a seasonal model may provide motivation for its employment over constant vaccination, contrary to what is stated by Onyango and M\"uller \cite{optimalvac}.  

\begin{figure}[t!]
\subfigure[][]{\label{fig:comparea}\includegraphics[width=9cm]{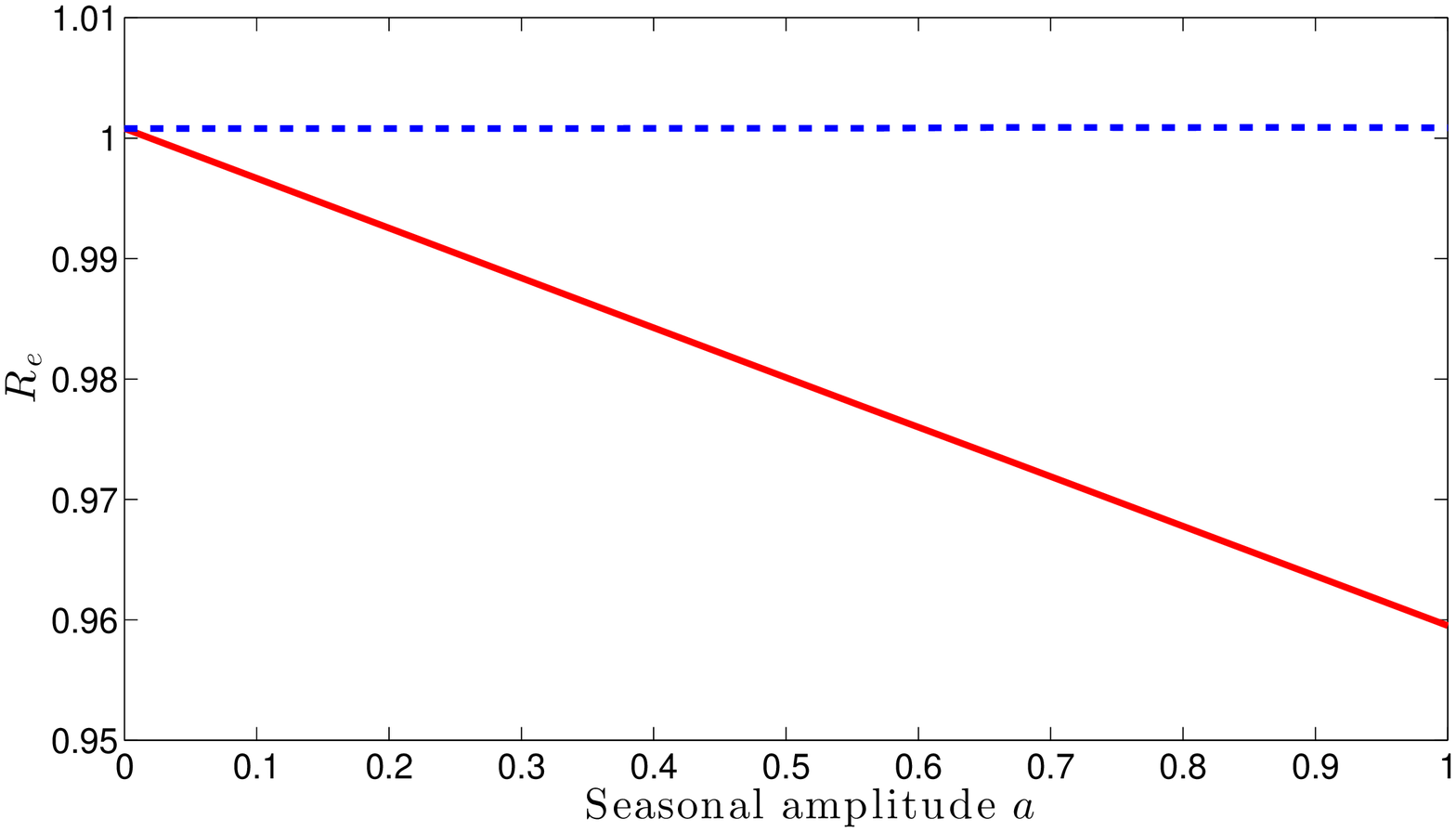} }
\subfigure[][]{\label{fig:compareb}\includegraphics[width=9cm]{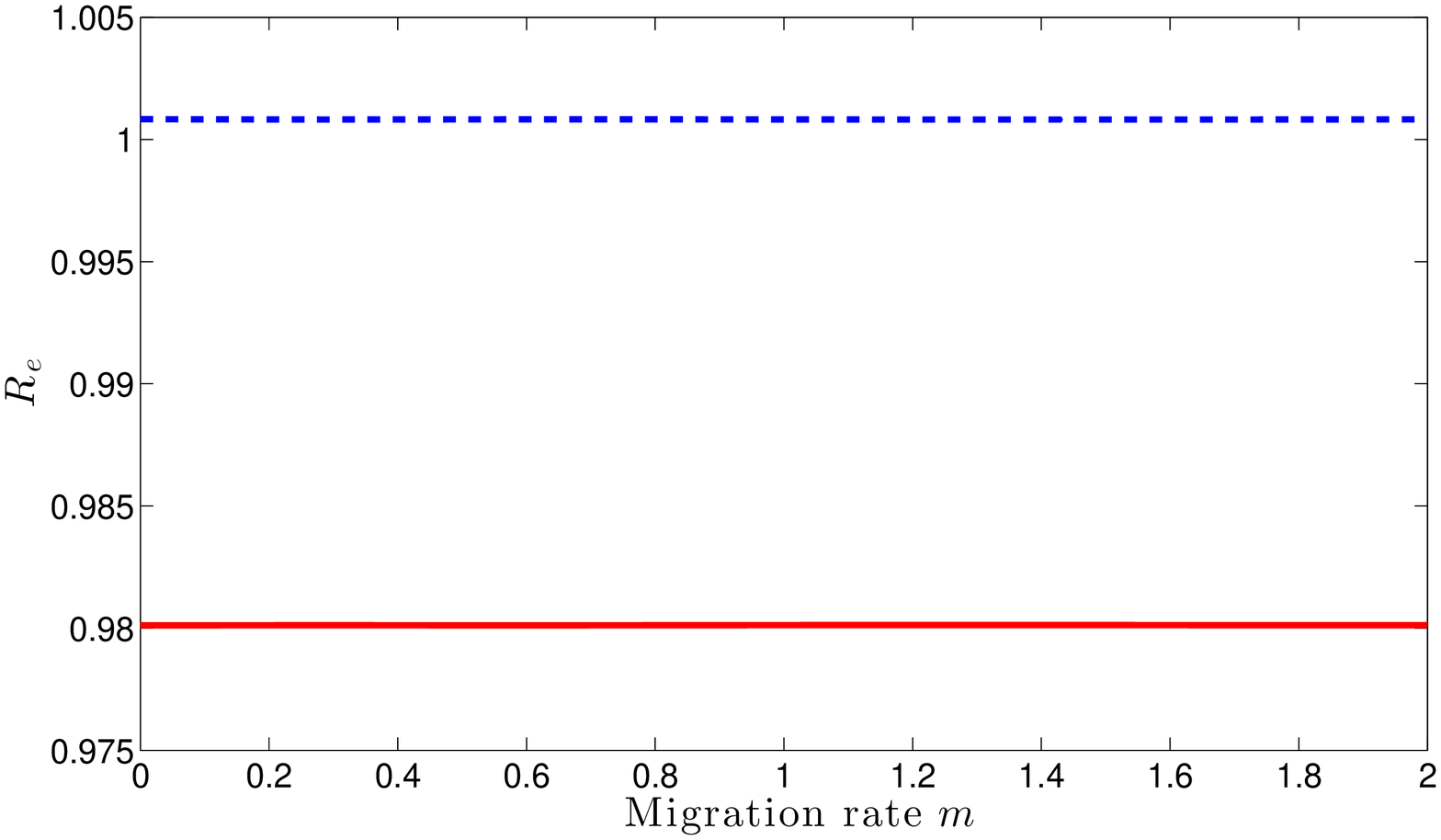}}
\caption{Simulations comparing $R_e$ for pulse vaccination (solid) and constant vaccination (dashed) of equal cost for the seasonal version of model (\ref{2patch}) with $\beta(t)=\beta(1+a\sin(2\pi (t-\theta)))$ and $\phi=0$, so the synchronised pulses occur just  before the high-transmission season.  The other parameters are $f=0$ and $\psi=0.23$, with the remainder as in Table \ref{table}. (a) $R_e$ vs amplitude of seasonality, $a$, with $m=0.0$ (b) $R_e$ vs $m$ with $a=0.5$.}
  \label{fig:compare}
  \end{figure}

%smeg

\subsection{Environmental transmission with identical patches and no seasonality}

Finally, we consider how environmental transmission affects the results.  To begin this section, we state a general theorem about the effective reproduction number for the autonomous (unpulsed) version of the general model (\ref{mod}) with environmental transmission.  The following theorem states that the effective reproduction number for the autonomous version of the general model (\ref{mod}) with environmental transmission is identical to the effective reproduction number of the autonomous model (\ref{mod}) without environmental transmission, but with the redefined direct transmission parameter $\widetilde{\beta}_{ij}=\beta_{ij}+\frac{\xi_j}{\nu_j}\epsilon_{ij}$.  

\begin{figure}[t!]
\subfigure[][]{\label{fig:enva}\includegraphics[width=9cm]{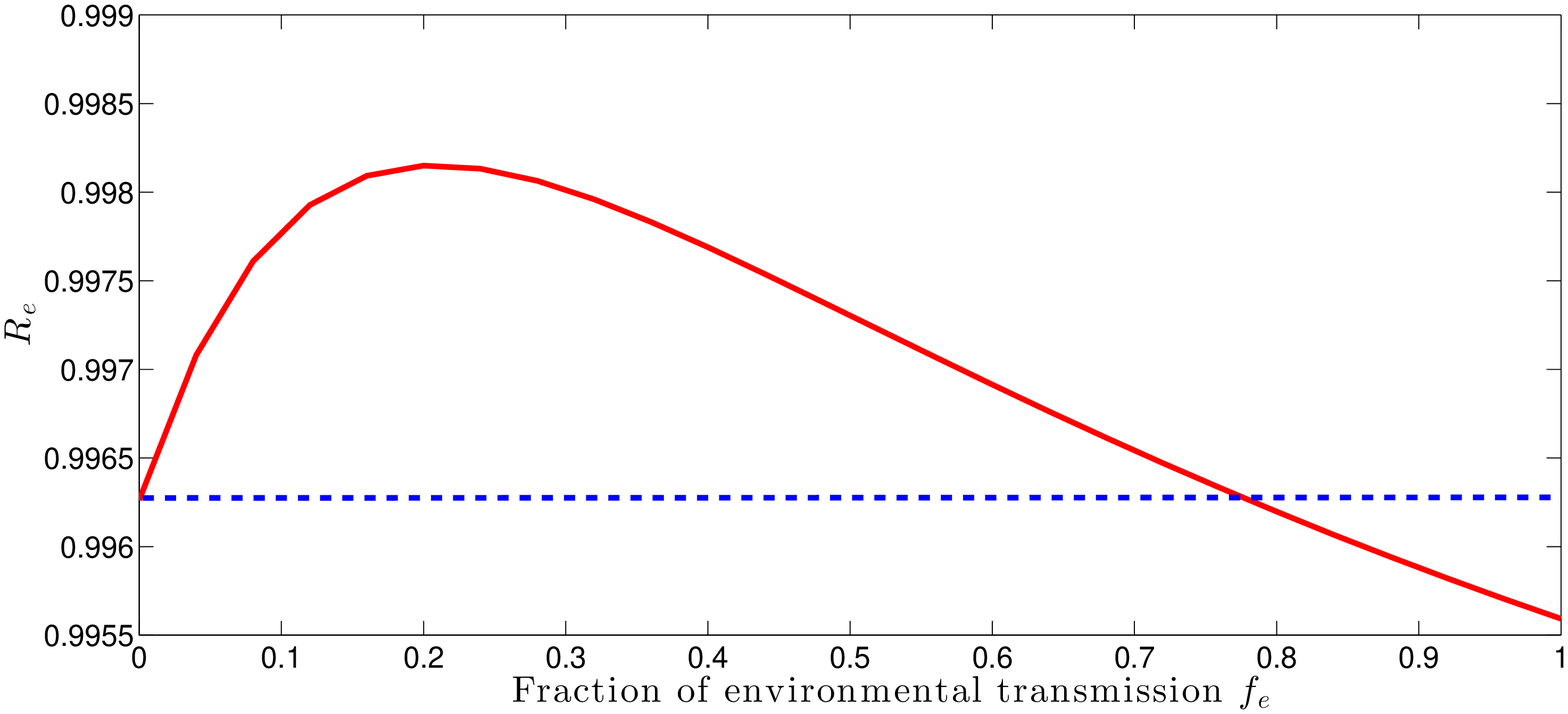}}
\subfigure[][]{\label{fig:envb}\includegraphics[width=9cm]{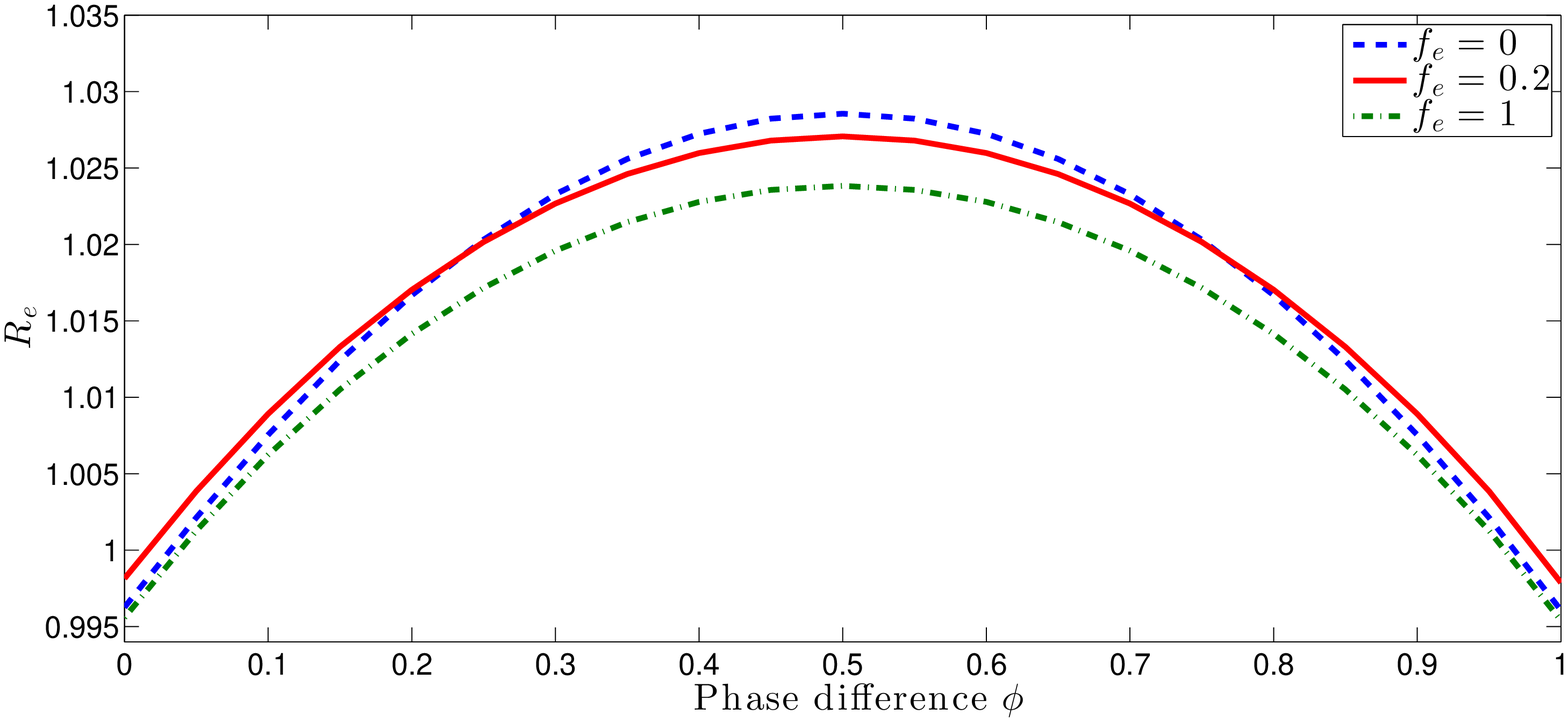} }
\caption{a) Simulations of the system with environmental transmission \eqref{envmodel} comparing $R_e$ for pulse vaccination (solid) and constant vaccination (dashed) of equal cost for the single-patch version of model (\ref{mod}) as the fraction of environmental transmission, $f_e$ increases from 0 to 1.  The parameters are as in Table \ref{table} with $\tilde\beta=319.655$ days$^{-1}$, $\nu=5 \ {\rm year}^{-1}$, $m=f=e_c=0$.  Note that the parameter values for $\xi$ and $\epsilon$ are absorbed into $f_e$ through a rescaling.  b)  $R_e$ vs the phase difference, $\phi$, between the vaccination pulses for three different values of $f_e$ (0, 0.2 and 1), when $m=1,f=c_e=0.1$.  The rest of the parameters are as in the previous simulation.}
  \label{fig:env}
\end{figure}

\begin{theorem}\label{envR0}
Denote $\widehat{R}_e$ as the effective reproduction number of the autonomous version (\ref{mod}).  Let $\widetilde{\beta}_{ij}=\beta_{ij}+\frac{\xi_j}{\nu_j}\epsilon_{ij}$ and $\widetilde{R}_e$ denote the effective reproduction number of the autonomous version of the multi-patch SIR sub-model (no environmental transmission) in (\ref{mod}) with the direct transmission parameter as $\widetilde{\beta}_{ij}$.   Then $\widehat{R}_e=\widetilde{R}_e$.
\end{theorem}
\begin{proof}
To find the reproduction number, $\widehat{R}_e$, for the autonomous version of (\ref{mod}), we utilize the standard next-generation approach \cite{vanD}.  Then the infection component linearization at the disease-free equilibrium is $\dot{x}=(F-V)x$, where the $2N\times 2N$ matrices $F$ and $V$ can be written in the block-triangular form:
\begin{align*}
 F&=\begin{pmatrix} D & E \\ 0 & 0 \end{pmatrix}, \quad V=\begin{pmatrix} A & 0 \\ B & C \end{pmatrix},
 \end{align*}
 in which $D,E,A,B,C$ are $N\times N$ matrices.  Here $B$ and $C$ are diagonal matrices with $\xi_i$ and $\nu_i$ ($i=1,\dots n$) as the respective diagonal entries. The entries of matrices $D$, $E$ and $A$ are as follows:  $D_{i,j}=\beta_{ij}\overline{S}_i, \ E_{i,j}=\epsilon_{ij}\overline{S}_i$ and $ A_{i,j}=(\mu_i+\gamma_i+k_{ii})\delta_{ij} + (1-\delta_{ij})(-k_{ij})$, where $\delta_{ij}$ is the Kronecker delta function and $\overline{S}=\left(\overline{S}_1,\dots,\overline{S}_N\right)$ is the disease-free  equilibrium. 
Then
$$ V^{-1}=\begin{pmatrix} A^{-1} & 0 \\ C^{-1}BA^{-1} & C^{-1} \end{pmatrix} \quad \text{and so} \ \ FV^{-1}=\begin{pmatrix}  DA^{-1} +EC^{-1}BA^{-1} & *  \\ 0 & 0 \end{pmatrix}. $$
$\widehat{R}_e$ is the spectral radius of $FV^{-1}$, so  $\widehat{R}_e=\rho(FV^{-1})= \rho\left( (D+EC^{-1}B)A^{-1} \right)$.  Now define $\widetilde{\beta}_{ij}=\beta_{ij}+\frac{\xi_j}{\nu_j}\epsilon_{ij}$ and consider the effective reproduction number, $\widetilde{R}_e$, of the autonomous version of (\ref{mod}) with no environmental transmission, but with direct transmission rate $\widetilde{\beta}_{ij}$.  It is not hard to see that $\widetilde{R}_e=\rho\left( (D+EC^{-1}B)A^{-1} \right)$.  Thus $\widehat{R}_e=\widetilde{R}_e$, and the result is obtained.
\end{proof}

Thus, for the autonomous case, the addition of environmental transmission to an SIR metapopulation model, by considering the system (\ref{mod}), does not qualitatively affect the effective reproduction number.  We should note that environmental transmission can result in a substantive delay in epidemic onset and its duration of first peak when compared to the analogous regime of direct transmission \cite{BTW11}, so the nature of the transient dynamics is affected by environmental transmission.

For simulations, we include the environmental parameters in the two-patch model (\ref{2patch}) and suppose the the patches are identical.  Then, using the notation from the previous identical patch case, we can write the infected-component equations as: 
\begin{align*}
\frac{dI_i}{dt}&=\tilde\beta S_i((1-f_e)((1-f)I_i+fI_j)+f_e((1-c_e)G_i+c_eG_j))-(\mu+\gamma)I_i-mI_i+mI_j  \tag{18} \label{envmodel} \\
\frac{dG_i}{dt}&=\xi I_i -\nu G_i,  \qquad i\neq j,\ i,j=1,2 ,
\end{align*}
where $\tilde{\beta}\equiv \beta+\frac{\xi}{\nu}\epsilon$ is the total transmission rate, $\beta\equiv \beta_{11}+\beta_{12}$, $\epsilon\equiv \epsilon_{11}+\epsilon_{12}$, $f_e\equiv \frac{\xi\epsilon}{\nu\tilde\beta}$ is the fraction of environmental transmission, $f\equiv \frac{\beta_{12}}{\beta}$ and $c_e\equiv \frac{\epsilon_{12}}{\epsilon}$ is fraction of cross-transmission for direct and environmental transmission, respectively.  Then, by Theorem \ref{envR0}, the effective reproduction number for the autonomous model with constant per capita vaccination rate of susceptibles, $\sigma$, is 
\begin{align*}
\widehat{R}_e=\frac{\widetilde{\beta} b}{(\mu+\sigma)(\mu+\gamma)}.
\end{align*}
Clearly, adding environmental transmission to the identical two-patch model (\ref{2patch}) does not alter the autonomous (without pulse vaccination) patch reproduction number $\widehat{R}_e$ if we re-define the transmission rate in (\ref{r0hat}) to be $\widetilde{\beta}$.

However, when pulse vaccination is introduced into the model, we find
that the fraction of environmental transmission, $f_e$, affects $R_e$.
In Figure \ref{fig:enva}, synchronous pulse vaccinations are compared
to constant vaccination as $f_e$ is varied.  The reproduction
number, $R_e$, under the pulse vaccination shows non-monotone behaviour
with respect to $f_e$, with a maximum occurring around $f_e=0.2$ and
the minimum occurring at $f_e=1$ (where all of the transmission is due
to the environment).  We remark that this graph looks the same for
many different values that we utilised for the migration rate and
cross transmission, in particular for the case of isolated patches.
The other parameters are as in Table \ref{table} with
$\tilde\beta=319.655$ days$^{-1}$ and $\nu=5 \ {\rm year}^{-1}$.  Note that the
parameter values for $\xi$ and $\epsilon$ are absorbed into $f_e$
through a rescaling.  Of course, we know from before that when
$f_e=0$, the vaccination strategies yield identical $R_e$ (proven in
the isolated patch case without seasonality).  In contrast, even in
the single patch case, environmental transmission can cause
disagreement in the reproduction numbers for pulse- and constant-vaccination strategies of equal cost.  This result can be viewed as an
impulsive analogue to that showing that sinusoidal transmission alters
the reproduction number for an SEIR model, but leaves the reproduction
number for the SIR model the same as with constant transmission
\cite{nicolas, browne}. Indeed, an SEIR model can be seen as a special
case of the no-impulse environmental transmission model.  In Figure
\ref{fig:envb}, we vary the phase difference, $\phi$, between the
vaccination pulses for three different values of $f_e$ (0, 0.2 and
1), when $m=1,f=c_e=0.1$.  The remaining parameters are as in Table \ref{table}.

  \section{Discussion} \label{sec7}
  
We have studied pulse vaccination in metapopulations, using poliovirus vaccination as our focus.  By allowing each patch to have distinct, periodic pulse-vaccination schedules connected with a common period --- along with considering seasonality, environmental transmission and two types of mobility --- we add more generality and complexity to prior models.  The effective reproduction number, $R_e$, is defined for the model, system (\ref{mod}), and found to be a global threshold.  If $R_e<1$, then the disease dies out; on the other hand, when $R_e>1$, the disease uniformly persists.
  
  Through theoretical analysis and numerical simulations, we were able to gain insights into optimising vaccination strategies in the metapopulation setting.  Theorem \ref{R0theorem} and the supporting numerical simulations show that synchronising vaccination pulses among connected patches is key in minimising the effective reproduction number.  An open problem is to analytically prove that synchronising the pulses minimises $R_e$ under more general conditions than are assumed in Theorem \ref{R0theorem}.  
  
Evidence from the epidemiological data suggests that pulse synchronisation 
 at different spatial scales  influences the effectiveness of a vaccination campaign.  Based on field studies, the WHO recommends that the duration of the vaccination campaign, in the form of national immunisation days, be as short as possible (1--2 days)  \cite{NID}.  The importance of administering the vaccine across a whole country in 1--2 days, as opposed to taking a longer period of time, may be in part due to the higher levels of synchronisation for the shorter duration vaccination campaign.  Increased seroconversion rates also seem to play a role in the optimality of mass vaccinations with short duration \cite{NID}.  On an international scale, the effectiveness of Operation MECACAR \cite{cdc} and a study of the effect of migration on measles incidence after mass vaccination in Burkina Faso \cite{measles} point to the importance of pulse synchronisation.  Our study highlights the critical role that the WHO and national governments can play in optimising disease control by synchronising mass vaccination campaigns among countries and regions.
  
  Another important problem is comparing the effectiveness of periodic mass (pulse) vaccination versus routine (constant) vaccination.  Disease-control authorities must consider certain logistical aspects, which may affect the cost of implementing a particular strategy.  From a mathematical perspective, the fundamental starting point for comparison is to consider $R_e$ for strategies of equal vaccinations per period.  For the case of no seasonality and environmental transmission, we find some  cases where the strategies are equivalent in terms of $R_e$.  When seasonality is included, a well-timed pulse-vaccination strategy (simultaneous pulses administered during the season before the high-transmission season) is optimal (assuming the patches have synchronous seasons), similar to results for the single-patch SIR model \cite{optimalvac}.  Future work will consider comparing pulse vaccination and constant-vaccination strategies in a stochastic model, which yields some insights not seen in the deterministic setting.

More work needs to be done in the case of environmental transmission.  When indirect transmission was considered to be a major mode of transmission in other studies,  a delay in epidemic onset  and its extension when compared to the analogous regime of direct transmission were observed \cite{BTW11}. This could be explained by the persistence of the virus in the environment leading to new infections generated over a longer duration than that of the direct contact. Such a two-step mechanism, with a human-to-environment segment and an environment-to-human segment, could lead to delay and extension of the effective infectious period when compared to that of direct human-to-human transmission \cite{BTW11}.  Interestingly, we found that varying the fraction of environmental transmission in the system alters the effective reproduction number $R_e$ under pulse vaccination, contrary to results for the case of constant-rate vaccination.  Further consideration of the interaction of environment-induced delay with the influence of seasonality on environmental transmission and pulse vaccination is the subject of our ongoing work. 

  Finally, we mention the importance of incorporating mobility and spatial structure into disease models.  In addition to our findings about the how mobility induces an advantage to synchronise pulse vaccination, population movement has other implications for disease control.  As found in previous work on autonomous models \cite{Smith?Li}, imbalance in migration rates among the patches can have a large effect on the overall reproduction number, which may alter the optimal vaccine distribution among patches or may influence disease-control strategies related to movement restriction.  The combination of population movement with complexities of control strategies and disease transmission presents many problems for which mathematical modelling may yield valuable insight.

\section*{Acknowledgements}

We are grateful to two anonymous reviewers whose comments greatly improved the manuscript. RJS?\ is supported by an NSERC Discovery Grant. For citation purposes, note that the question mark in ``Smith?" is part of his name.

\FloatBarrier

\end{document}